%% file: ieee_sp_main.tex
\documentclass[conference]{IEEEtran}
\IEEEoverridecommandlockouts
\usepackage{cite}
\usepackage{amsmath,amssymb,amsfonts}
\usepackage{algorithmic}
\usepackage{textcomp}
\usepackage{xcolor}
\def\BibTeX{{\rm B\kern-.05em{\sc i\kern-.025em b}\kern-.08em
    T\kern-.1667em\lower.7ex\hbox{E}\kern-.125emX}}
    
\usepackage{floatrow}
\newfloatcommand{capbtabbox}{table}[][\FBwidth] 
    
\usepackage{tikz}
\usepackage{amsmath}
\usepackage{amsthm}

\usepackage{soul,color}
\usepackage{etoolbox}
\usepackage{wrapfig}
\usepackage{wrapfig,lipsum,booktabs}
\usepackage{amsfonts}

\usepackage{mathtools}
\usepackage{graphicx}
\usepackage{makecell}
\usepackage[ruled,linesnumbered,vlined]{algorithm2e}
\usepackage{siunitx,  
            booktabs, 
            caption} 
            
\usepackage{amssymb,graphicx}

\usepackage{adjustbox}

\usepackage{multirow}
\usepackage{color, colortbl}
\definecolor{Gray}{gray}{0.9}
\definecolor{Red}{rgb}{1,0.7,0.7 }
\definecolor{Blue}{rgb}{0.6,0.9,0.97 }
\definecolor{Green}{rgb}{0.55,0.92,0.55 }

\makeatletter
\newcommand\bigDiamond{\mathop{\mathpalette\bigDi@mond\relax}}
\newcommand\bigDi@mond[2]{%
  \vcenter{\hbox{\m@th
    \scalebox{\ifx#1\displaystyle 2\else1.2\fi}{$#1\Diamond$}%
  }}%
}
\newcommand\bigLozenge{\mathop{\mathpalette\bigL@zenge\relax}}
\newcommand\bigL@zenge[2]{%
  \vcenter{\hbox{\m@th
    \scalebox{\ifx#1\displaystyle 2\else1.2\fi}{$#1\blacklozenge$}%
  }}%
}
\makeatother

\newtheorem{theorem}{Theorem}

\newtheorem{corollary}{Corollary}
\newtheorem{definition}{Definition}
\newtheorem{prop}{Proposition}
\newtheorem{remark}{Remark}

\DeclareMathOperator{\R}{\mathbf{R}}

\DeclarePairedDelimiter{\floor}{\lfloor}{\rfloor}

\begin{document}


\title{Privacy-Preserving Database Fingerprinting} 

\author{\IEEEauthorblockN{Tianxi Ji}
\IEEEauthorblockA{\textit{Dept.  of ECSE} \\
\textit{Case Western Reserve Univ.}\\
txj116@case.edu}
\and
\IEEEauthorblockN{Erman Ayday}
\IEEEauthorblockA{\textit{Dept.   of CSDS} \\
\textit{Case Western Reserve Univ.}\\
exa208@case.edu}
\and
\IEEEauthorblockN{Emre Yilmaz}
\IEEEauthorblockA{\textit{Dept.  of CSET} \\
\textit{Univ. of Houston-Downtown}\\
yilmaze@uhd.edu}
\and
\IEEEauthorblockN{Pan Li}
\IEEEauthorblockA{\textit{Dept.  of ECSE} \\
\textit{Case Western Reserve Univ.}\\
lipan@case.edu}
}

\maketitle
\thispagestyle{plain}
\pagestyle{plain}
\begin{abstract}

When sharing sensitive relational databases   with other parties, a database owner aims to (i) have privacy guarantees for the entries in the shared database, (ii) have liability guarantees (e.g., via fingerprinting) in case of unauthorized sharing of its database by the recipients, and (iii) provide a high  quality (utility) database to the recipients. 
We observe that sharing a relational database  with privacy and liability guarantees  
are orthogonal objectives. 
The former can be achieved by injecting noise into the database to  prevent inference of the original data values, whereas, the latter can be achieved by hiding unique marks inside the database to trace malicious parties (data recipients) who redistribute the data without the authorization. 
In this paper, we achieve these two objectives simultaneously by proposing a novel  entry-level differentially-private fingerprinting mechanism for relational databases. 

At a high level, the proposed mechanism fulfills the privacy and liability requirements by   leveraging the randomization nature that is intrinsic to fingerprinting and achieves desired entry-level privacy guarantees. 
To be more specific, we devise a bit-level random response scheme to achieve differential privacy guarantee for arbitrary data entries when sharing the entire database, and then, based on this, we develop an $\epsilon$-entry-level differentially-private fingerprinting mechanism. Next, 
we theoretically analyze the relationships between privacy guarantee, fingerprint robustness, and database utility by  deriving closed form expressions. The outcome of this analysis allows us to  bound the privacy leakage caused by attribute inference attack and characterize  the privacy-utility coupling and privacy-fingerprint robustness coupling. Furthermore, we also propose a sparse vector technique (SVT)-based solution to control the cumulative privacy loss when fingerprinted copies of a database are shared with multiple recipients. 

We experimentally show that our proposed mechanism achieves stronger fingerprint robustness than the state-of-the-art mechanisms (e.g., the fingerprint cannot be compromised even if a malicious recipient changes 80\% of the entries in the fingerprinted database), and higher database utility compared to the simple composition of database perturbation under local  differential privacy followed by fingerprinting (e.g., the statistical utility of the shared database by the proposed scheme is more than $10\times$ higher than  perturbation followed by fingerprinting). 
\end{abstract}

\begin{IEEEkeywords}
database, privacy, liability, fingerprinting
\end{IEEEkeywords}

\input{sections/introduction}

\input{sections/relatedwork}

\input{sections/Preliminaries}


\input{sections/proposed_method}
\input{sections/relationships}
\input{sections/multiple_query_new_version}

\input{sections/experiment}

\input{sections/discussion}

\input{sections/conclusion}

\typeout{}
\bibliographystyle{plain}
\bibliography{sample-base.bib}

\input{sections/appendix}

\end{document}

%% file: sections/introduction.tex
\section{Introduction}
\label{sec:intro}

Massive data collection and availability of relational databases (collection of data records with the same set of attributes \cite{codd2002relational}) 
are very common in the current big data era. This   results in an increasing demand to share such   databases  with (or among) different service providers (SPs), such as companies, research institutions, or hospitals, for the purpose of ``do-it-yourself'' calculations, like personal advertisements, social recommendations, and customized healthcare.

Most relational databases include personal data, and hence they usually contain sensitive and proprietary information, e.g., medical records collected as part of an agreement which restricts redistribution. This poses three major challenges in the course of database sharing with different SPs: (i) the database owner is obligated to  protect the privacy of data entries in the shared database to comply with the privacy policy and ensure confidentiality, (ii) the database owner needs to prevent illegal redistribution of the shared databases, and eventually should be able to prosecute  the malicious SPs who leak its data, 
and (iii) the shared database needs to maintain high utility to support accurate data mining and data analysis.

\subsection{Current Status} \label{sec:current}
Many works have attempted to address the    challenges on privacy and liability in isolation. To address the privacy challenge, various  data sanitization  metrics are proposed, e.g., $k$-anonymity \cite{sweeney2002k}, $l$-diversity \cite{machanavajjhala2007diversity}, $t$-closeness \cite{li2007t}, and differential privacy \cite{dwork2014algorithmic}. Among them, differential privacy has been developed  as a \textit{de facto} standard for responding to statistical queries from databases with provable privacy guarantees. It can also be used  to share personal data streams or an entire database (i.e., identity query) in a privacy-preserving manner \cite{chanyaswad2018mvg,ji2021differentially}. Differentially-private mechanisms hide the presence or absence of a data record in the database by perturbing the query results with noise calibrated to the query sensitivity. 

To protect copyright and deter illegal redistribution, different database watermarking and fingerprinting mechanisms are devised to prove database ownership (i.e., identification of the database owner from the shared database)~\cite{agrawal2003watermarking,sion2004rights} and database possession (i.e., differentiating between the SPs who have received copies of the database) \cite{li2005fingerprinting,ji2021curse,yilmaz2020collusion}. In particular, when sharing a database with a specific SP, the database owner embeds a unique fingerprint (a binary string customized for the data recipient SP) in the database. The embedded fingerprint is typically hard to be located and removed even if a malicious SP attacks the fingerprinted database (to identify and distort the fingerprint). 
 
In the literature, only a few works have attempted to combine database  sanitization and fingerprinting during database sharing. In particular, \cite{gambs2018entwining,bertino2005privacy,kieseberg2014algorithm}  propose inserting fingerprints into  databases  sanitized using $k$-anonymity and \cite{schrittwieser2011algorithm} proposed embedding fingerprints into a database   sanitized by the $(\alpha,\beta)$-privacy model \cite{rastogi2007boundary}. 
However, these works solve the aforementioned challenges in a two-stage (sequential) manner, where data sanitization is conducted followed by fingerprinting. As a result, they end up changing a significant amount of entries in the database and they significantly compromise the utility of the shared database. We experimentally corroborate this in Section \ref{sec:experiment}.  
These works also do not address the critical problem of controlling cumulative privacy loss if the same database is repeatedly shared with multiple SPs.

\subsection{This Paper}\label{sec:pics}

In this work, we bring together data sanitization and fingerprinting in a unified data sharing algorithm, consider a stronger privacy model compared to previous works, and develop entry-level  differentially-private fingerprinting for relational database sharing. In what follows, we briefly summarize the main contributions and insights of our work,  and discuss its limitation caused by a unique requirement of DBMS  (Database  Management  System)  design.


\noindent\textbf{Contributions.} We observe that database fingerprinting is a   randomized mechanism (which essentially performs bitwise randomization), and thus is naturally endowed with certain level of privacy protection. However, this hidden property (privacy protection) is ignored in the literature. In this paper, we harness this intrinsic randomness  and   transform it into a provable privacy guarantee. In particular, 
\begin{itemize}
\item We propose a bit-level random response scheme, which fingerprints   insignificant bits of data entries using pseudorandomly generated binary mark bits, to achieve $\epsilon$-entry-level differential privacy (Definition \ref{def_dp}) for the entire database. Then, we devise an   $\epsilon$-entry-level differentially-private fingerprinting mechanism (Algorithm \ref{algo:fingerprint_sp_n}) based on the   bit-level random response scheme.

 \item We establish a comprehensive and solid theoretical foundation (Section \ref{sec:connections}) to quantify  the properties of the proposed $\epsilon$-entry-level differentially-private fingerprinting mechanism from 3 dimensions: (i)  the privacy guarantee of it under attribute inference attack, (ii) the fingerprint robustness of the mechanism when it is subject to various attacks that can be launched against a fingerprinting mechanism, and (iii)  the relationship among privacy and utility,  and privacy and fingerprint robustness.
 
  
\item We devise a sparse vector technique (SVT)-based solution (Algorithm \ref{algo:MultipleFingerprintedDatabasesReleasing}) to control the cumulative privacy loss when fingerprinted databases are shared with multiple SPs.
    
\item We evaluate and validate the proposed mechanism using a real-life   database. Our experimental results show that  the proposed 
  mechanism 
  (i) provides higher fingerprint robustness than a state-of-the-art database  fingerprinting mechanism, (ii) provides higher database utility than the na\"ive solution, which first perturbs the database under  local  differential privacy guarantee, and then inserts the fingerprint, and (iii) achieves high database utility  along with high privacy guarantee and fingerprint robustness even when the database is shared for multiple times. 
\end{itemize}

\noindent\textbf{Insights.} This paper is the first to show  the  feasibility of   considering  privacy  and  liability in a unified mechanism to protect data privacy and prevent unauthorized data redistribution simultaneously. 
The proposed mechanisms can be used to guide a database owner to (i) generate    privacy-preserving fingerprinted databases based on customized requirements on utility, privacy level, and fingerprint robustness and (ii) assess the privacy leakage under multiple database sharings and  set the privacy budget accordingly in each sharing. 

By harnessing the intrinsic randomness in fingerprinting to bridge privacy and liability, we  believe that our  work will  draw  attention  to  other  challenges  and  urgent  research problems along this direction including mitigating  both  attribute inference attacks and  fingerprinting distortion  that utilizes the correlations among data entries, addressing the collusion attack lunched by allied malicious SPs, and improving the database utility via Bayesian denoising that leverages the data distributions. In Section \ref{sec:discussion}, we provide potential solutions to each of these open problems.


\noindent\textbf{Limitations.} In this work, we consider the sharing of entire relational databases, where each data record (i.e., row) can be uniquely identified by a immutable pseudo-identifier (i.e., the primary key). 
This is a unique requirement of DBMS  (Database  Management  System)  design, and hence in this work, we do not consider membership inference attacks as they become irrelevant under these settings.  
We further discuss this in detail in Section \ref{sec:system_model}.

\noindent\textbf{Paper organization.} 
In Section~\ref{sec:related_work} we review related works in the literature, which is followed by 
the privacy, system, and threat models 
in Section \ref{sec:system_model}.  In Section \ref{sec:dp_fp}, we present   the proposed entry-level differentially-private fingerprinting mechanism. Then, we theoretically investigate the relationships between database utility, fingerprint robustness, and privacy guarantees in Section~\ref{sec:connections}. We develop the SVT-based mechanism to share multiple fingerprinted databases under entry-level differential privacy in Section~\ref{sec:multiple_sharing}.  We   evaluate the proposed scheme via extensive experiments in Section~\ref{sec:experiment}.  We provide further discussions and point out   open problems with potential solutions in Section \ref{sec:discussion}. Finally, Section~\ref{sec:conclusion} concludes the paper.

%% file: sections/relatedwork.tex
\section{Related Work}
\label{sec:related_work}
Quite a few works have studied the problem of protecting data privacy  and ensuring  liability in isolation when sharing databases. 
The works that are closest to ours include~\cite{gambs2018entwining,bertino2005privacy,kieseberg2014algorithm,schrittwieser2011algorithm}. Yet, these works embed watermark or fingerprint into an already sanitized database, instead of considering data sanitization and marking together (as a unified process). Thus,  such sequential  processing of database will result in significant degradation in utility. To be more specific,  Bertino et al. \cite{bertino2005privacy} adopted the binning method \cite{lin2002using} to generalize the database first, and then watermark the binned data to protect copyright. Kieseberg et al. \cite{kieseberg2014algorithm} and Schrittwieser et al. \cite{schrittwieser2011algorithm} proposed fingerprinting a database generalized by $k$-anonymity. Gambs et al. \cite{gambs2018entwining} sanitized the database using the $(\alpha,\beta)$-privacy model \cite{rastogi2007boundary}, which selects a true data record from the domain of the database with probability $\alpha$ and includes a fake data record that is not in the domain of the database with probability $\beta$,  and then they embed personalized fingerprint in the database. These studies usually change a large amount of data entries, which degrades database utility. 
In particular, it has been observed that $k$-anonymity may create data records that leak information due to the 
lack of diversity in some sensitive attributes, and it  does not protect against attacks based on background knowledge \cite{ganta2008composition}. 
Note that we do not include a comparison   with these mechanisms, because their adopted privacy models cannot be transformed into differential privacy. However, in Section \ref{sec:experiment}, we compare the proposed scheme with the approach that first perturbs the database entries under local differential privacy, and then   fingerprints the result using a state-of-the-art database fingerprinting scheme.

Our work is different from the previous mechanisms, because we consider sanitization and fingerprinting as a unified process (to have higher data utility), we aim at achieving provable privacy guarantees during fingerprint insertion, we adopt a privacy model that is customized for relational database systems, i.e., entry-level differential privacy, and we establish a connection among fingerprint robustness, data privacy, and data utility in a coherent way.

%% file: sections/Preliminaries.tex


\section{Privacy, System, and Threat Models}\label{sec:system_model}

Now, we discuss the considered privacy model, database fingerprinting system, and various threats. First, we review the definition of a relational database and its unique features, which are important for our specific choice of privacy model.

\begin{definition}[Relational database \cite{codd2002relational}] A relational database   denoted as $\R$ is a  collection of $T$-tuples. Each of these tuples represents a data record containing $T$ ordered attributes. Each data record is also associated with a primary key, which is used to uniquely identify that record. We denote the $i$th data record in $\R$ as $\mathbf{r}_i$ and the primary key of $\mathbf{r}_i$ as $\mathbf{r}_i.PmyKey$.
\end{definition}

\noindent\textbf{Unique Features of Relational Database.} In order to support database operations, such as union, intersection, and update, the primary keys should \textbf{not} be changed   if a database is fingerprinted or pirated  \cite{li2005fingerprinting,agrawal2003watermarking,li2003constructing}.\footnote{In Database Management System (DBMS) design, the primary keys are required to be immutable, as updating a primary key can   lead to the update of potentially many other tables or rows in the system. The reason is that in DBMS, a primary key may also serve as a foreign key (a column that creates a relationship between two tables in DBMS). For instance, consider the database in Section \ref{sec:experiment} in which each data record represents a school applicant. Here, the primary key of the data record can be chosen as the applicant's unique identification number. The identification number can then be used to refer to another table keeping the applicant's real name, email, etc. Thus, the update of the primary key causes further changes in other tables.}  
 Due to the uniqueness and immutability of the primary keys in relational databases, the presence of a specific data record is not a private information in general. In other words, it is no secret  whether an individual's data record (a specific $T$-tuple) is present in a database or not. 
Hence, the common definition of neighboring databases (which differ by one row) in the differential privacy literature does not apply in the case of sharing relational databases. Thus, we consider an alternative definition of neighboring relational databases as follows.

\begin{definition}[Neighboring  relational databases]\label{def:neighboring}
 Two   relational databases $\R$ and $\R'$ are called  neighboring databases, 
 if they only differ by   \textbf{one entry}, i.e., an attribute of one individual's data.
\label{def:nei_db}
\end{definition}
Given Definition \ref{def:neighboring}, we have the sensitivity of a relational database defined as follows.  

\begin{definition}
[Sensitivity of  a relational database]
Given a pair of neighboring relational databases $\R$ and $\R'$ that differ   by one entry (e.g., the $t$th attribute of the $i$th row, $\mathbf{r}_{i}[t]$ and $\mathbf{r}_{i}[t]'$), 
the    sensitivity  is defined as $  \Delta = \sup_{\R,\R'}||\R - \R'||_F = \sup_{\mathbf{r}_{i}[t],\mathbf{r}_{i}[t]'}\Big|\mathbf{r}_{i}[t] -  \mathbf{r}_{i}[t]'\Big|$.
\label{def:bit_sensitivity}
\end{definition}

\subsection{Privacy Model}\label{sec:privacy_model}
Next, we provide our privacy model customized specifically for relational databases with immutable primary keys (i.e., data record pseudo-identifiers).

\begin{definition}[$\epsilon$-entry-level differential privacy] 
\label{def_dp}
A randomized mechanism $\mathcal{M}$ with domain $\mathcal{D}$ satisfies $\epsilon$-entry-level differential privacy if for any two neighboring relational databases $\R, \R' \in \mathcal{D}$, and for all $\mathcal{S} \in \rm{Range}(\mathcal{M})$, it holds that $\Pr[\mathcal{M}(\R)= S] \leq e^{\epsilon}\Pr[\mathcal{M}(\R') =  S]$, where $\epsilon>0$. 
\end{definition}

\begin{remark}
Our proposed privacy model (Definition \ref{def_dp}) is adapted from the conventional notation of $\epsilon$-differential privacy \cite{dwork2014algorithmic}, which obfuscates the presence or absence of an entire row in $\R$. Since the database recipient can easily identify if an individual is present in $\R$ by directly checking its primary key, the conventional $\epsilon$-differential privacy is \textbf{not appropriate} in the setting we consider. In contrast, our privacy model, which aims at  obscuring the specific value of an arbitrary entry in $\R$, better suits the requirement of database management system (DBMS) design. As discussed in Section \ref{sec:pics},     destroying       pseudo-identifiers   to   prevent linkability or membership inference attacks becomes an ill-posed problem for our considered case of DBMS. Thus, we will focus on the attribute inference attacks instead of membership inference attacks  in the paper.
\end{remark}

Similar to the conventional differential privacy, a   relaxation of Definition \ref{def_dp} is the $(\epsilon,\delta)$-entry-level differential privacy, i.e.,  $\Pr[\mathcal{M}(\R)= S] \leq e^{\epsilon}\Pr[\mathcal{M}(\R') =  S]+\delta$, where $\delta\in[0,1]$.  In Section \ref{sec:multiple_sharing}, we will analyze the achieved $(\epsilon,\delta)$-entry-level differential privacy guarantee  when the  proposed  privacy-preserving  fingerprint scheme is executed multiple times for repeated database sharing using the sparse vector technique.  Table \ref{table:notations} in  Appendix \ref{sec:notations} lists the frequently used notations.

\subsection{Database Fingerprinting System}\label{sec:system}

We present the  system model  in Figure \ref{fig:model}. We consider a  database owner with a  relational database  denoted as $\mathbf{R}$, who wants to share it with at most $C$ SPs (e.g., to receive services, to help researchers, or for collaborative research). 
To prevent unauthorized redistribution of the database by a malicious SP (e.g., the $c$th SP in Figure \ref{fig:model}), the database owner includes unique fingerprints in all shared copies of the database. 
The fingerprint  essentially changes different entries in $\mathbf{R}$ at different positions (indicated by the yellow dots). 
The fingerprint bit-string customized for the $c$th SP  ($\mathrm{SP}_c$) is denoted as $f_{\mathrm{SP}_c}$, and the database received by $\mathrm{SP}_c$ is represented as $\widetilde{\R}_c$. 
Both $f_{\mathrm{SP}_c}$ and $\widetilde{\mathbf{R}}_c$ are obtained  using the proposed mechanism  discussed in Section \ref{sec:dp_fp}.   We   let $\widetilde{\R}$   represent a general instance of the privacy-preserving fingerprinted database.

\begin{figure}[htb]
  \begin{center}
     \includegraphics[width= 0.98\columnwidth]{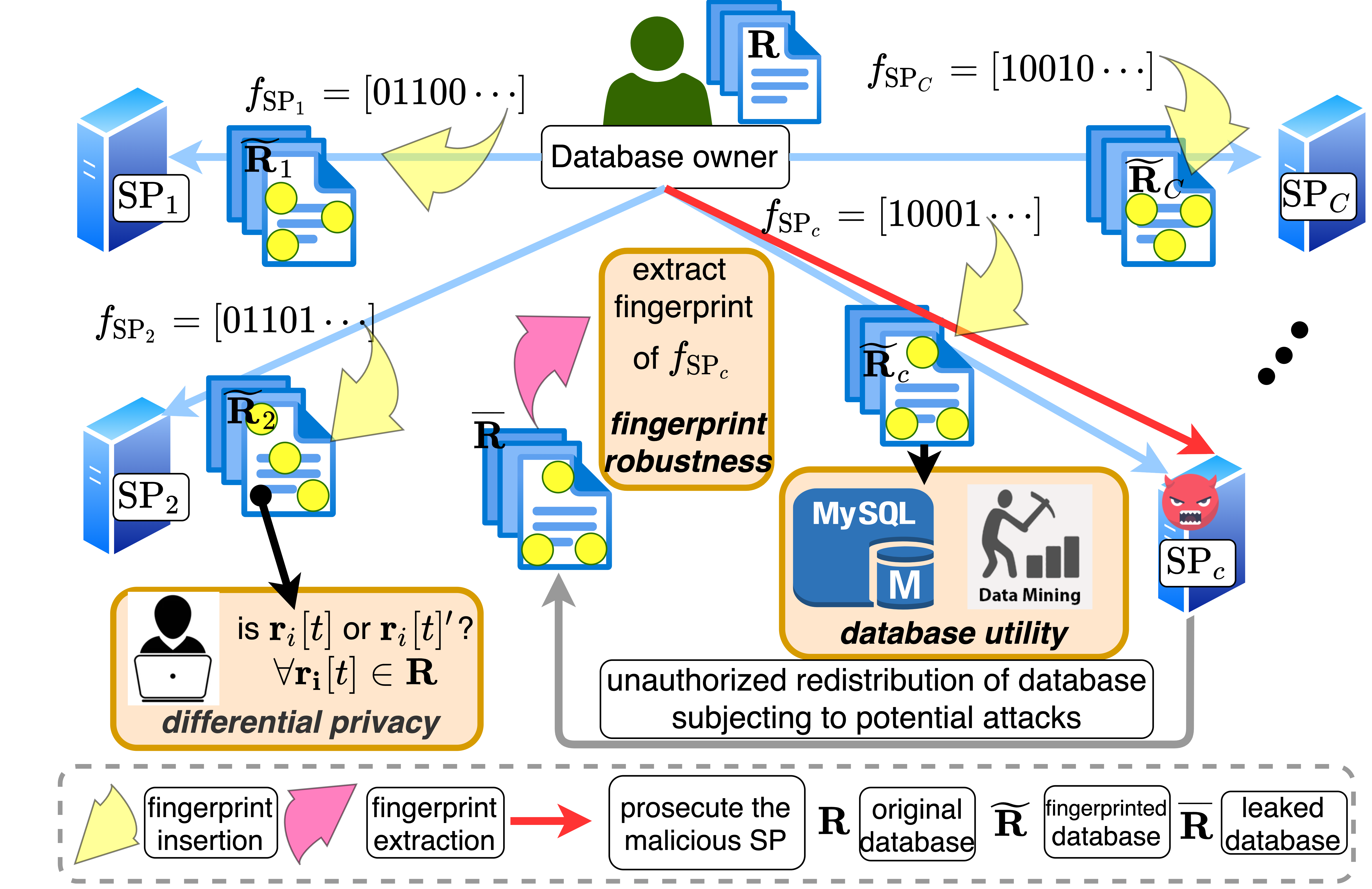}
      \end{center}
  \caption{\label{fig:model}  System model  considered in this paper. All shared copies of the database meet entry-level differential privacy, fingerprint robustness, and database utility requirements.
}
\end{figure}


We aim to achieve three main goals in this system, i.e., 
(i)  Privacy guarantee for each data entry, i.e., a data analyst cannot distinguish between $\mathbf{r}_i[t]$ and $\mathbf{r}_i[t]'$ by inferring its received copy $\widetilde{\R}$. 
  (ii)      High fingerprint robustness in order to successfully extract a malicious SP's fingerprint (even if the malicious SP tries to distort the fingerprint to mitigate detection) if the database is redistributed without authorization.
(iii)  High data utility for the fingerprinted database in order to support accurate database queries and data mining tasks.


\subsection{Potential Threats}\label{sec:threats}


Since we consider developing a mechanism to simultaneously achieve data  privacy and liability guarantees, we also need to address the corresponding threats   from these  two aspects. In particular, the malicious SP can 
\begin{itemize}
    \item  Adopt sophisticated learning     methods to infer the original values of each data entry (in the shared database)  by using its prior knowledge or other revealed data entries. In particular, we consider an adversary who knows all data entries except for one, and it will use advanced learning methods to infer the original value of the unknown data entry. In Section \ref{sec:discussion}, we discuss how to augment our proposed mechanism to defend against attribute inference attacks that leverage the correlations among data entries.

    \item Conduct various attacks to distort the  embedded fingerprint bit-strings, e.g., random bit flipping attack (which changes bits of randomly selected data entries), subset attack (which randomly removes fractions of data records from the database), and correlation attack (which changes the data entries   by taking advantage of data correlations available from other resources). In Section \ref{sec:connections}, we discuss these attacks in detail and derive closed-form robustness  expression achieved by our mechanism. 
\end{itemize}

%% file: sections/proposed_method.tex

\section{Privacy-preserving Fingerprinting}
\label{sec:dp_fp}

In this section, we present the proposed privacy-preserving fingerprinting mechanism. First, we discuss the design principles of the mechanism. Next, we develop a general condition for  bit-level random response to achieve entry-level  differentially-private database release (or sharing). Then, we propose a concrete mechanism built upon  the  bit-level random response scheme to achieve provable privacy guarantees during fingerprint insertion.

\subsection{Principles of Mechanism Design}\label{sec:design}

First of all, we notice that all fingerprinting schemes achieve liability guarantees on each shared copy by flipping different collections  of randomly selected   bits of the data entries using a certain probability. The collections of selected bits vary for different SPs and their fingerprinted values are determined by the unique   fingerprint  bit-strings of the SPs. Thus, we observe that database fingerprinting schemes are randomized mechanisms, which essentially 
perform  bitwise-randomization, i.e., change the data values  by introducing noise at the bit-level of the data entries instead of directly perturbing the data entries (i.e., introducing noise at the enter-level). As a consequence, we also establish our entry-level differentially-private  fingerprinting scheme by conducting bitwise-randomization, and to achieve  provable privacy guarantee, we calibrate the flipping probability based on the sensitivities of the data entries.

On the contrary, the entrywise-randomization adopted by the conventional differentially-private output perturbation mechanisms, e.g., \cite{chanyaswad2018mvg,ji2021differentially}, are  infeasible as a building block of a fingerprinting mechanism, because they substantially  change all data entries by adding noises drawn from some probability distributions. Although local differential privacy via randomized response only changes each data entry with a particular probability determined by the privacy budget and the number of possible instance of the entries \cite{bassily2020practical}, connecting this probability with the randomly generated fingerprint bit-string is not straightforward. This is because randomly changing each bit of each data entry (by fingerprinting)  may not lead to the identical random effect required by local differential privacy.  Hence, it is also not suitable for  designing a database fingerprinting mechanism.



Based on the principle of bitwise-randomization, we consider achieving entry-level differential privacy for the release of the entire database by using bit-level random response. In particular, when sharing a database with a specific SP, the values of    selected   insignificant bits of selected  data entries are determined  by XORing   them with  random binary variables, which is different for different data sharing instances (with various SPs). Such modification of  bit positions in the database using different binary values  can also be considered as inserting different fingerprints, which can be used to accuse a malicious SP if there is a data leakage.  Moreover, to achieve high utility for the shared database, we simultaneously achieve entry-level differential privacy and fingerprinting, instead of fingerprinting a differentially-private database (or sanitizing an already fingerprinted database). This will be further  elaborated in Section \ref{sec:connections} and validated in Section \ref{sec:experiment}.




\subsection{Privacy-preserving Database Release via Bit-level Randomization}

Traditional differential privacy guarantees that the released 
statistics computed from a database (such as mean or histogram) are independent of the absence or presence of an individual. However, in this work we consider the release (sharing) of the entire fingerprinted  database and the existence of a particular individual can be easily determined by checking its corresponding primary key in the released copy (as discussed in Section \ref{sec:system_model}).  Therefore, in this work, we consider the privacy of database entries (i.e., attributes of individuals). Many works have also studied the problem of achieving differential privacy for entries in a database, e.g., \cite{chaudhuri2006random,pejo2020sok,ji2021differentially,johnson2018towards}. Our work differs from them as  we can achieve additional liability guarantees by   introducing randomness at bit level and uniting fingerprinting with   privacy.

Now, we   define the   bit-level random response scheme. 

\begin{definition}
[Bit-level random response scheme] A bit-level random response scheme (pseudorandomly) selects some bits of some data entries in the database and changes the bit values of such entries by conducting an XOR operation on them with independently  generated random  binary mark bits, denoted as $B$, where $B\sim\mathrm{Bernoulli}(p)$.
\label{def:fp_mech}
\end{definition}

Existing database  fingerprinting schemes only mark the insignificant bits of the data entries to introduce  tolerable error in the database. In this paper, we assume that  the $k$th to the last bit of an entry is its $k$th  insignificant bit.  If the  $k$th  insignificant bit of attribute $t$ of data record $\mathbf{r}_i$ (represented  as  $\mathbf{r}_i[t,k]$) is selected, then the bit-level random response scheme changes its value as $\mathbf{r}_i[t,k]\oplus B$, where $\oplus$ stands for the XOR operator, and $B$ is a Bernoulli random  variable with parameter $p$.

As discussed in Section \ref{sec:system}, the modified database should have small amount of error, because both database owner and data recipient SPs expect high utility for the shared (and received) relational databases. Thus, to obtain high database utility, we let the bit-level random response scheme only change the last $K$ bits of data entries, i.e., $k\in[1,K]$. As a result, we can develop the following condition for such a scheme to achieve $\epsilon$-entry-level differential privacy on the entire database. The proof is deferred to Appendix \ref{sec:proof_thm_dpfp}.

\begin{theorem}\label{thm:dpfp}
Given a relational database $\R$ with  sensitivity $\Delta$ (Definition \ref{def:bit_sensitivity}), 
a bit-level random response scheme, which only changes the last $K$ bits of data entries, satisfies $\epsilon$-entry-level differential privacy if $K = \floor[]{\log_2\Delta}+1$ and $p\geq \frac{1}{e^{\epsilon/K}+1}$.
\end{theorem}

\subsection{An  $\epsilon$-Entry-level Differentially-Private Fingerprinting Mechanism}\label{sec:an_instant}


Due to the randomness involved in the bit-level random response scheme, for any given $p$ and $\R$, the output databases will vary for each different run. However, in order to   detect the guilty SP who is responsible for the leakage of a database, it is required that the fingerprinted database shared with a specific SP must be unique and it can be reproduced by the database owner  even if the mark bits, i.e., $B$'s, are generated  randomly. 
In this section, we discuss how to develop an instantiation of an $\epsilon$-entry-level differentially-private fingerprinting mechanism based on the bit-level random response scheme, i.e., a mechanism that satisfies Theorem \ref{thm:dpfp} and at the same time, is reproducible when sharing a fingerprinted copy with  any  specific SP  using  a given Bernoulli distribution parameter  $p$.






First, we collect all fingerprintable bits in $\R$, i.e., all   insignificant bits (the last $K$ bits) of   all entries, in a set $\mathcal{P}$:\\
$\mathcal{P} = \Big\{\mathbf{r}_i[t,k]\Big|i\in[1,N],t\in[1,T],k\in[1,\min\{K,K_t\}]\Big\}$, 
where $N$ is the number of data record in $\R$, and $K_t$ represents the number of bits to encode the $t$th attribute in $\R$. 
When the database owner wants to share a fingerprinted copy of $\R$ with an SP with a publicly known external ID denoted as  $ID_{\mathrm{external}}$, it first generates an internal ID for this SP denoted as  $ID_{\mathrm{internal}}$. We will elaborate the generation of $ID_{\mathrm{internal}}$ in Section \ref{sec:inter_step}. Then, the database owner generates the unique fingerprint for this SP using a cryptographic hash function, i.e., $\mathbf{f} = Hash(\mathcal{Y}|ID_{\mathrm{internal}})$, where $\mathcal{Y}$ is the secret key of the database owner and $|$ represents the concatenation operator. We use $L$ to denote the length of the generated fingerprint.\footnote{We use MD5 to generate a 128-bits fingerprint string, because if the database owner shares $C$ copies of its database, then as long as $L\geq \ln C$, the fingerprinting mechanism can thwart exhaustive 
search and various types of attacks. Usually, a 64-bits
fingerprint   can provide high robustness \cite{li2005fingerprinting}. }

The database owner also has  a cryptographic pseudorandom sequence generator $\mathcal{U}$, which   selects the data entries and their insignificant bits, and determines the mask bit $x$ and  fingerprint bit $f$ (which is an element of the fingerprint bit-string $\mathbf{f}$) to obtain the Bernoulli random   variable (i.e., $B = x\oplus f$). To be more specific, for each $\mathbf{r}_i[t,k]$ in $\mathcal{P}$,  the database owner sets the  initial seed  as $s = \{\mathcal{Y}|\mathbf{r}_i.PmyKey|t|k\}$. If $\mathcal{U}_1(s)\ \mathrm{mod}\ \floor{\frac{1}{2p}} =0$ ($p= \frac{1}{e^{\epsilon/K}+1}$), then $\mathbf{r}_i[t,k]$ is fingerprinted. Next, the database owner decides the value of mask bit $x$ by checking if  $\mathcal{U}_2(s)$ is even or odd, and sets the fingerprint index $l = \mathcal{U}_3(s)\ \mathrm{mod}\ L$. It obtains the mark bit $B = x \oplus \mathbf{f}(l)$, and finally changes the bit value of $\mathbf{r}_i[t,k]$ with $\mathbf{r}_i[t,k]\oplus B$. We summarize the   steps to generate a fingerprinted database $\mathcal{M}(\R)$ for SP $ID_{\mathrm{external}}$ in Algorithm \ref{algo:fingerprint_sp_n}.

\begin{algorithm}
\small
\SetKwInOut{Input}{Input}
\SetKwInOut{Output}{Output}
\Input{The original   database $\mathbf{R}$, the privacy budget $\epsilon$, number of changeable bits $K$, the Bernoulli distribution parameter  $p= \frac{1}{e^{\epsilon/K}+1}$, database owner's secret key $\mathcal{Y}$, and pseudorandom number sequence generator $\mathcal{U}$. 
}
\Output{\resizebox{0.95\hsize}{!}{$\epsilon$-differentially-private fingerprinted database $\mathcal{M}(\R)$.}}

Construct the fingerprintable set $\mathcal{P}$.

Generate the internal ID, i.e.,  $ID_{\mathrm{internal}}$ for this SP (will be elaborated in Section \ref{sec:inter_step}).\label{line:set_ID}

\resizebox{1\hsize}{!}{Generate the fingerprint   string, 
i.e., $\mathbf{f} = Hash(\mathcal{Y}|ID_{\mathrm{internal}})$.}\label{line:generate_fp}


\ForAll{$\mathbf{r}_i[t,k]\in\mathcal{P}$}{

Set pseudorandom seed $s = \{\mathcal{Y}|\mathbf{r}_i.PmyKey|t|k\}$, 


\If{$\mathcal{U}_1(s)\ \mathrm{mod}\  \floor{\frac{1}{2p}} = 0$}{

\resizebox{1\hsize}{!}{Set mask bit $x = 0$, if $\mathcal{U}_2(s)$ is even; otherwise    $x=1$.}

Set fingerprint index $l=\mathcal{U}_3(s)\ \mathrm{mod}\ L$.

Let fingerprint bit $f = \mathbf{f}(l)$.

Obtain mark bit $B = x\oplus f$.

Set $\mathbf{r}_i[t,k]= \mathbf{r}_i[t,k]\oplus B$. \algorithmiccomment{insert fingerprint}


}

}
 Return the fingerprinted database $\mathcal{M}(\R)$.

		\caption{Generate $\mathcal{M}(\R)$ for SP $ID_{\mathrm{external}}$.}
\label{algo:fingerprint_sp_n}
\end{algorithm}


\begin{theorem}\label{corollary:dpfp_share}
Algorithm \ref{algo:fingerprint_sp_n} is $\epsilon$-entry-level differentially-private.
\end{theorem}

\begin{remark}Please refer to Appendix \ref{sec:dpfp_share_proof} for the proof. 
The proposed database fingerprinting scheme is different from the existing ones discussed in Section \ref{sec:related_work}, as all existing schemes fingerprint each selected bit by replacing it with a new value obtained from the XOR of pseudorandomly generated mask bit $x$ and fingerprint bit $f$. As a result, the new value is independent of the original bit value in the relational database.  This is why the privacy guarantees of existing fingerprinting schemes cannot be explicitly analyzed. On the contrary,   we fingerprint each selected bit by XORing it with a Bernoulli random variable $B$ to make the fingerprinted entries dependent on the original bit value in the relational database. This enables us to derive a tight upper bound on the ratio of the probabilities of a pair of neighboring databases returning identical fingerprinted  outcomes, which is the key step to further connect this bound to a provable privacy guarantee.
\end{remark}

Note that $\mathcal{U}$ produces a sequence of random numbers using an initial seed, and it is computationally prohibitive to compute the  next random number in the sequence without knowing the seed.  As a consequence, from an SP's point of view, the results of $\mathcal{M}(\R)$'s are random. However,  $\mathcal{M}(\R)$ can be reproduced by the database owner who has access to the private key of itself and the external and internal IDs of  SPs.



\subsection{Extracting Fingerprint from a Leaked Database}\label{sec:extract_fp_leaked}

When the database owner observes a leaked (or pirated) database denoted as $\overline{\R}$, it will try to identify the traitor (malicious SP) by extracting the fingerprint from $\overline{\R}$ and comparing it with the fingerprints of SPs who have received its database before.

We present the fingerprint extraction procedure  in Algorithm \ref{algo:extract} in Appendix \ref{sec:appendix_extraction}. Specifically, the database owner first initiates a fingerprint template $(f_1,f_2,\cdots,f_L) = (?,?,\cdots,?)$.  Here, ``$?$'' means that the fingerprint bit at that position remains to be determined.\footnote{Similar symbol has also been used in other works   \cite{boneh1995collusion,li2005fingerprinting,agrawal2003watermarking,ji2021curse}.} Then, the database owner locates the positions of the fingerprinted bits exactly as in Algorithm \ref{algo:fingerprint_sp_n}, and fills in each ``?'' using majority voting.  To be more precise, it first constructs the  fingerprintable sets $\overline{\mathcal{P}}$ from   $\overline{\R}$, i.e., $\overline{\mathcal{P}} = \Big\{\overline{\mathbf{r}_i}[t,k]\Big|i\in[1,\overline{N}],t\in[1,T],k\in[1,\min\{K,K_t\}]\Big\}$, where $\overline{\mathbf{r}_i}[t,k]$ is the $k$th  insignificant bit  of attribute $t$ of the $i$th data record in $\overline{\R}$, and  $\overline{N}$ is the number of records in $\overline{\R}$. 
Note that $\overline{N}$ may not be equal to $N$, because a malicious SP may conduct the subset attack (as will be discussed in Section \ref{sec:subset}) to remove some data records from the received database before leaking it.  Second, the database owner selects the same bit positions, mask bit $x$,  and fingerprint index $l$ using the pseudorandom seed $s = \{\mathcal{Y}|\mathbf{r}_i.PmyKey|t|k\}$. Third, it recovers the mark bit as $B = \overline{\mathbf{r}_i}[t,k] \oplus \mathbf{r}_i[t,k]$ and fingerprint bit at index $l$ as $f_l = x\oplus B$. Since the value of $f_l$ may be changed by the malicious SP, the database owner maintains and updates two counting arrays $\mathbf{c}_0$ and $\mathbf{c}_1$, where $\mathbf{c}_0(l)$ and $\mathbf{c}_1(l)$ record the number of times $f_l$ is recovered as 0 and 1, respectively. Finally, the database owner sets $\mathbf{f}(l) = 1$, if $\mathbf{c}_1(l)>\mathbf{c}_0(l)$, otherwise $\mathbf{f}(l) = 0$.  The database owner compares the constructed fingerprint bit-string with the fingerprint customized for each SP who has received the database, and one of these SPs will be considered as guilty if there is a large overlap between its fingerprint and the constructed one. Note that the database owner can accuse at least one SP certainly \cite{li2005fingerprinting}. It has been shown that the database owner can correctly identify the malicious SP as long as the overlapping between fingerprints is above 50\% \cite{ji2021curse}.

%% file: sections/relationships.tex
\section{Associating Privacy, Fingerprint Robustness, and Database Utility}\label{sec:connections}
In the previous section, we have presented a mechanism that achieves  provable  privacy guarantee when fingerprinting a   database. Here, we investigate its impact on the database utility and fingerprint robustness, and also establish the connection between $p$ (the Bernoulli distribution parameter, which represents the probability of  changing one insignificant bit of a data entry),\footnote{Note that $p$ is the probability of the mark bit $B$ taking   value 1. The probability of the mark bit taking the value 0 is also $p$ by design (see Line 10 in Algorithm \ref{algo:fingerprint_sp_n}). Thus, the probability of a specific bit position being fingerprinted is $2p$ (see Line 6 in Algorithm \ref{algo:fingerprint_sp_n}). We would like to remind that $2p<1$, since the probability of a specific bit position is not selected to be marked is $1-2p$. Please also see Proof of Theorem \ref{corollary:dpfp_share} for more details.}  
entry-level differential privacy guarantee ($\epsilon$),  fingerprint robustness, and utility of   shared databases. 
 We show a graphical relationships between these in Figure \ref{fig:relationships}, where the arrow means ``leads to". Clearly, we can obtain the high-level conclusion that differential privacy and fingerprint robustness are not conflicting objectives that can be achieved at the same time, whereas, at the cost of database utility.

\begin{figure}[htb]
  \begin{center}
     \includegraphics[width= 0.9\columnwidth]{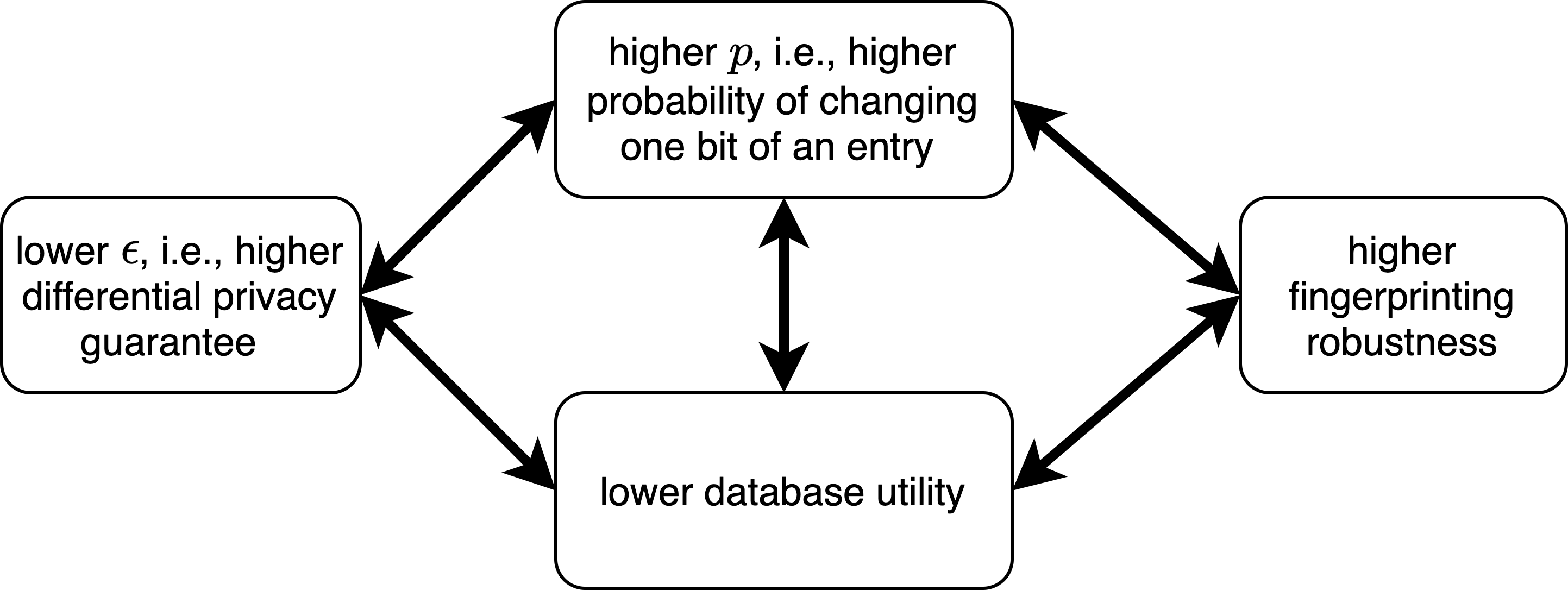}
      \end{center}
  \caption{\label{fig:relationships} 
  Relationship among $p$ (the probability of changing one insignificant bit of a data entry), differential privacy guarantee ($\epsilon$),   fingerprint robustness, and database utility. 
}
\end{figure}

\subsection{Privacy against General Attribute Inference Attacks}\label{sec:privacy_against_att_inf}

After receiving the fingerprinted database $\mathcal{M}(\R)$, a malicious SP can leverage sophisticated  learning methods to infer the original value of each data entry. In this section we show that under our proposed privacy model (i.e., entry-level differential privacy), the malicious SP's inference capability can never exceed a certain threshold.

In particular, we consider a   malicious SP who has access to $\R_{/\boldsymbol{r}_i[t]}$ (i.e., the original values of all data entries except the $t$th attribute of the $i$th data record, $\boldsymbol{r}_i[t]$), and its inference capability (denoted as $\mathrm{InfCap}$) is defined as 
\begin{equation}
    \mathrm{InfCap} = 
    \Pr(\boldsymbol{r}_i[t] = \zeta_1|\mathcal{M}(\R),\R_{/\boldsymbol{r}_i[t]}),
    \label{eq:InfCap}
\end{equation}
which is the posterior probability of the unknown entry $\boldsymbol{r}_i[t]$ takes a specific value $\zeta_1$. $\mathrm{InfCap}$   covers a wide-range of inference attacks using   learning-based techniques, because most of the learning frameworks give outputs in terms  of posterior probabilities, e.g.,  Bayesian inference and deep learning.  In particular, we can reach to the following proposition  about the inference capability of a malicious SP (the proof is shown in Appendix \ref{sec:proof_att_inf_bound}).

\begin{prop}\label{prop:att_inf_bound}
No matter what learning-based inference attack the malicious SP conducts, its inference capability can never be higher than $\frac{ \psi e^\epsilon}{ \psi e^\epsilon +1}$, i.e., $\mathrm{InfCap}\leq  \frac{ \psi e^\epsilon}{ \psi e^\epsilon +1}$, where $\psi = \frac{\Pr(\boldsymbol{r}_i[t]=\zeta_1  |\R_{/\boldsymbol{r}_i[t]})}{\Pr(\boldsymbol{r}_i[t]=\zeta_2|    \R_{/\boldsymbol{r}_i[t]})}$ is the   malicious SP's prior knowledge on the ratio between the probabilities of the unknown entry $\boldsymbol{r}_i[t]$ taking different values (i.e., $\zeta_1$ and $\zeta_2$) given all other entries are known. 
\end{prop}


For a given $\psi$, $\frac{ \psi e^\epsilon}{ \psi e^\epsilon +1}$ decreases as $\epsilon$ decreases, it means that the higher the entry-level differential privacy guarantee (smaller $\epsilon$) the lower the inference capability of  malicious SPs.  In Appendix \ref{sec:app_att_inf_atk}, we consider a similar  adversary proposed in \cite{liu2016dependence}, empirically investigate its inference capability, and compare it with our derived upper bound.    
The above considered   adversary who  knows the entire database except one data entry is a standard threat model in the differential privacy literature. 
In some of the real-world attacks, an adversary can also utilize some publicly known auxiliary information (e.g., correlation among data entries \cite{liu2016dependence,yilmaz2021genomic,chanyaswad2018mvg,yang2015bayesian} or social connections \cite{kifer2011no,ji2021differentially}) to improve its inference capability. 

Some works have attempted to augment the conventional differentially-private mechanisms to also make them robust against inference attacks utilizing the auxiliary information. For example, Liu et al. \cite{liu2016dependence} propose to augment the Laplace
mechanism with a dependence coefficient, which computes the
query sensitivity of correlated data. Chanyaswad et al. \cite{chanyaswad2018mvg} reinforce the Gaussian mechanism by considering the row- and column-wise covariance matrix in database. 
Ji et  al. \cite{ji2021differentially} enhance random response of binary data by modeling the correlation between binary data using log-associations. 

Our adopted privacy model can also be easily augmented to achieve robustness against attribute inference attacks that use additional information about data correlations. In particular, if a malicious SP knows the pairwise data entry correlations (measured in terms of pairwise joint distributions  between columns), it can improve $\mathrm{InfCap}$  
by factoring $\psi$, 
i.e.,  $\psi = \frac{\Pr(\boldsymbol{r}_i[t]=\zeta_1  |\R_{/\boldsymbol{r}_i[t]})}{\Pr(\boldsymbol{r}_i[t]=\zeta_2|    \R_{/\boldsymbol{r}_i[t]})} = \frac{\Pr(\boldsymbol{r}_i[t]=\zeta_1, \R_{/\boldsymbol{r}_i[t]})}{\Pr(\boldsymbol{r}_i[t]=\zeta_2,    \R_{/\boldsymbol{r}_i[t]})} = \frac{\prod_{k\neq t}\Pr(\boldsymbol{r}_i[t]=\zeta_1, \boldsymbol{r}_i[k])}{\prod_{k\neq t} \Pr(\boldsymbol{r}_i[t]=\zeta_2,    \boldsymbol{r}_i[k])}$. Then, to achieve a provable privacy guarantee against attribute inference attacks using data correlations, 
the database owner can further augment the proposed entry-level differentially private mechanism by involving  
the auxiliary information of $\frac{\Pr(\boldsymbol{r}_i[t]=\zeta_1, \boldsymbol{r}_i[k])}{\Pr(\boldsymbol{r}_i[t]=\zeta_2,    \boldsymbol{r}_i[k])}$ in the mechanism design, i.e., adjusting $p$ accordingly based on the value of $\frac{\Pr(\boldsymbol{r}_i[t]=\zeta_1, \boldsymbol{r}_i[k])}{\Pr(\boldsymbol{r}_i[t]=\zeta_2,    \boldsymbol{r}_i[k])}$ and $\epsilon$. In Section \ref{sec:discussion}, we provide a mathematically viable approach to achieve such augmentation.



The   goal of 
attribute inference attack using data correlations is to compromise data privacy. 
In Section \ref{sec:robs_corr_atk}, we discuss a counterpart of it to compromise the fingerprint robustness; the malicious SP can also leverage the discrepancy between data correlations before and after fingerprint insertion to distort the embedded fingerprint bits. Also, in Section \ref{sec:discussion}, we provide a thorough discussion on how the database owner can also utilize data correlations to mitigate both attribute inference attack and correlation-based attacks target on fingerprints.

\subsection{Database Utility}

Fingerprinting naturally changes the content of the database, and
thus degrades the utility. Here, we evaluate the utility of  a fingerprinted
database from both data accuracy and data correlation perspectives: we quantify the impact of Algorithm \ref{algo:fingerprint_sp_n}  on the accuracy of each fingerprinted data entry and the joint   probability distribution of any pair of data entries.  The theoretical analyses are summarized in the following propositions.  These considered utility metrics are application independent, and in general, the higher the accuracy of data entries and  pairwise   joint distributions, the better the task-specific application utilities, e.g., classification accuracy and mean square error. We empirically validate this statement  by evaluating the proposed scheme by focusing on task-specific application utilities in Section \ref{sec:experiment}.
\begin{prop}\label{thm:cell_accuracy}
Let $\mathbf{r}_i[t]$ and $\widetilde{\mathbf{r}_i[t]}$ be the original and the fingerprinted values of the $t$th attribute of the $i$th row. Then, the expected error caused by fingerprinting, i.e.,  $\mathbb{E}_{B\sim\mathrm{Bernoulli}(p)}\Big[\Big|\mathbf{r}_i[t]-\widetilde{\mathbf{r}_i[t]}\Big|\Big]$, falls in the range of $[0, \Delta p]$,  where $\Delta$ is the sensitivity of a pair of  neighboring relations (see Definition \ref{def:bit_sensitivity}), and $p$ is the probability of a mark bit $B$ taking value 1, i.e., the probability of changing one insignificant bit of a data entry.
\end{prop}

The proof is in Appendix~\ref{sec:proof_cell_accuracy}. Clearly, the higher the value of $p$, the larger the expected absolute difference between a fingerprinted data entry and the original value. This suggests that  
the database owner can set the value of $p$ based on its  requirement of data entry accuracy when generating a fingerprinted database, which achieves a certain level of entry-level differential privacy, 
and vice versa.  This leads us to the following corollary.

\begin{corollary}\label{corollary:density}
Define \textbf{fingerprint density} as   $||\mathcal{M}(\R)-\R||_{1,1}$, where $||\cdot||_{1,1}$ is the matrix $(1,1)$-norm which sums over the absolute value of each entry in the matrix. Then, we have $\mathbb{E}_{B\sim\mathrm{Bernoulli}(p)}\Big[||\mathcal{M}(\R)-\R||_{1,1}\Big]\in [0,\Delta p NT]$.
\end{corollary}

In Section \ref{sec:multiple_shares}, we will exploit the concept of  fingerprint density to develop a SVT-based solution to share  fingerprinted databases with multiple  SPs.

\begin{prop}\label{thm:joint}
Let $\Pr(   \R[t]=\pi, \R[z] = \omega  )$ and $\Pr(  \widetilde{ \R[t]}=\pi, \widetilde{\R[z]} = \omega  )$ 
be the joint probability of the $t$th attribute   taking value $\pi$ and the $z$th attribute taking value $\omega$ before and after fingerprinting, respectively. Then,  
$\Pr(  \widetilde{ \R[t]}=\pi, \widetilde{\R[z]} = \omega  )$ falls in the range of 
$\resizebox{0.95\hsize}{!}{$\Big[  \Pr\Big(\R[t] = \pi,\R[z] = \omega\Big) (1-p)^{2K}+\Pr_{\min}(\R[t],\R[z])\Big(1-(1-p)^K    \Big)^2$}, \\
\resizebox{1\hsize}{!}{$\Pr\Big(\R[t] = \pi,\R[z] = \omega\Big) (1-p)^{2K}+\Pr_{\max}(\R[t],\R[z])\Big(1-(1-p)^K    \Big)^2     \Big]
$}$,\\
where $\Pr_{\min}(\R[t],\R[z])$ (or $\Pr_{\max}(\R[t],\R[z])$) denotes  the minimum (or maximum) joint probability of the $t$th and $z$th attributes in the original database.
\end{prop}

The proof is in Appendix~\ref{sec:proof_joint}. 
By marginalizing over $\R[t]$ and $\widetilde{\R[z]}$, we can have the following corollary about the impact of fingerprinting on the marginal distributions.

\begin{corollary}\label{corollary:marginal}
Let $\Pr(   \R[t]=\pi)$ and $\Pr(  \widetilde{ \R[t]}=\pi )$ be the marginal 
probability of the $t$th attribute   taking value $\pi$  before and after fingerprinting, respectively. Then, 
$\Pr(  \widetilde{ \R[t]}=\pi )$ belongs to 
$\resizebox{0.85\hsize}{!}{$\Big[  \Pr\Big(\R[t] = \pi\Big) (1-p)^{2K}+\Pr_{\min}(\R[t])\Big(1-(1-p)^K    \Big)^2$}, \\
\resizebox{0.85\hsize}{!}{$\Pr\Big(\R[t] = \pi \Big) (1-p)^{2K}+\Pr_{\max}(\R[t] )\Big(1-(1-p)^K    \Big)^2     \Big]
$}$, 
where $\Pr_{\min}(\R[t] )$ (or $\Pr_{\max}(\R[t] )$) denotes  the minimum (or maximum) marginal probability of the $t$th   attribute in the original database.
\end{corollary}

Clearly, when $p$ is small, both joint distributions and marginal distributions will be close to that of the original databases. This means that the fingerprinted database will have higher statistical utility for smaller values of $p$.

\subsection{Fingerprint Robustness}\label{sec:fp_robustness_connection}

Although, Li et al. \cite{li2005fingerprinting} have attempted to analyze fingerprint robustness by studying the false negative rate (i.e., the probability that the database owner   fails to extract the exact fingerprint from a pirated database), they do not establish the direct connection between the robustness and the tuning parameter (the fingerprinting ratio, which can be interpreted as  a counterpart of $p$ in our work) in their mechanism.

In this paper, we investigate the robustness of the proposed fingerprinting mechanism against three attacks, i.e., the random bit flipping attack \cite{agrawal2003watermarking,li2005fingerprinting,craver1998resolving}, subset attack \cite{li2005fingerprinting,bui2013robust,yilmaz2020collusion,craver1998resolving}, and correlation attack \cite{yilmaz2020collusion,ji2021curse}. 
 These attacks are all well-established attacks for database fingerprinting, and have been widely studied to investigate the robustness of a database fingerprinting mechanism in the literature. 
In the following, we quantitatively analyze the relationship between $p$ (the probability of changing one insignificant bit of a data entry) and fingerprint robustness against these three attacks.

\subsubsection{Robustness Against Random Flipping Attack}\label{sec:robust_rand_flip}
In random flipping attack, a malicious SP flips each of the $K$ last bits of data entries in $\widetilde{\R}$   with probability $\gamma_{\mathrm{rnd}}$ with the goal of distorting  the data in the fingerprinted   positions. In \cite{ji2021curse}, the authors have empirically shown that the malicious SP ends up being uniquely accusable as long as the extracted fingerprint   from the leaked database has more than $50\%$ matches with the malicious SP's fingerprint. 

However, as the database owner shares more fingerprinted copies of its database with different SPs, to uniquely hold the correct malicious SP responsible, it requires more bit matches between the extracted fingerprint and the malicious SP's fingerprint. Thus, the number of bit matches (denoted as $D$,   $D\leq L$) should be set based on the number of fingerprinted sharings of the same database. Please refer to Appendix \ref{sec:analysis_rnd} for the discussion of how to determine $D$ given the number of times a database is shared (with different SPs) and the fingerprint length $L$.

Given $D$, we evaluate the robustness of the proposed fingerprinting mechanism against random bit flipping  attack in terms of the probability (denoted as $P_{\mathrm{rbst\_rnd}}$) that the database owner successfully extracts   any $D$ fingerprint bits of the malicious SP. 
Let the $l$th bit of the fingerprint string be embedded $w_l$ times in $\widetilde{\R}$ (with the  probability  $(\frac{1}{L})^{w_l}$). Thus,   to extract this fingerprint bit correctly from a copy of $\widetilde{\R}$ that is compromised by the random bit flipping attack, the database owner needs to make sure that  at most $\floor{\frac{w_l}{2}}$ bits in $\widetilde{\R}$ that are marked by the $l$th bit of the fingerprint string are flipped by the malicious SP, which happens with probability $p_l = \sum_{q=0}^{\floor{\frac{w_l}{2}}}{w_l \choose q}\gamma_{\mathrm{rnd}}^q(1-\gamma_{\mathrm{rnd}})^{w_l-q}$. 

Let $m$ be the number of fingerprinted bit positions in the database ($m\leq NKT$) received by the malicious SP, and define set $\mathcal{W}$ as \resizebox{0.7\hsize}{!}{$\mathcal{W} = \{w_1,w_2,\cdots,w_L>0|\sum_{l=1}^Lw_l=m\}$}. Let also $\mathcal{L}_D$ be the collection of any $D$ bits of the malicious SP's fingerprint ($|\mathcal{L}_D|=D$). Then,  by marginalizing all possible instances of malicious SP's fingerprint,  all   collections of  $D$ fingerprint bits of it, and the  number of fingerprinted bits $m$,  we have the closed form expression of $P_{\mathrm{rbst\_rnd}}$ in terms of $p$  as  
$P_{\mathrm{rbst\_rnd}} = \sum_{m=1}^{NKT}\Big(\sum_{w_l\in\mathcal{W},l\in[1,L]}\sum_{\mathcal{L}_D}\prod_{l\in\mathcal{L}_D} p_l(\frac{1}{L})^{w_l}\Big){NKT \choose m}(2p)^m\\(1-2p)^{NKT-m}$. 
To show that higher $p$ leads to more robustness against the random bit flipping attack, it is equivalent to show that $P_{\mathrm{rbst\_rnd}}$ is monotonically increasing with $p$ ($0<p<0.5$)\footnote{A rational database owner will not change more than 50\% of the bit positions in a database, because it will significantly compromise the database utility, and a malicious SP can   flip the bits back and then launch an attack.}.  Detailed analysis is shown in  Appendix \ref{sec:analysis_rnd}.




\subsubsection{Robustness Against Subset Attack}
\label{sec:subset}
In subset attack, the malicious  SP generates a pirated database by selecting each data record in $\widetilde{\R}$ for inclusion (in the pirated database) with probability $\gamma_{\mathrm{sub}}$. This attack is shown to be much weaker than the random bit flipping attack \cite{li2005fingerprinting,ji2021curse,yilmaz2020collusion}. According to \cite{li2005fingerprinting} (page 40), the subset attack cannot succeed (i.e., distorting even one fingerprint bit) unless the malicious SP excludes all the rows fingerprinted by at least one fingerprint bit. 
Thus, we measure the robustness of the proposed fingerprinting mechanism against subset attack using the probability (denoted as $P_{\mathrm{rbst\_sub}}$) that the malicious SP fails to exclude all fingerprinted rows involving a particular fingerprint bit (note that our analysis can be easily generalized for excluding a fraction of fingerprinted rows). Since the probability that a specific row is fingerprinted by a specific fingerprint bit is $1-(1-p/L)^{KT}$,  we have the closed form expression of $P_{\mathrm{rbst\_sub}}$ in terms of $p$ (the probability of changing an insignificant bit of an entry) as 
\resizebox{1\hsize}{!}{$P_{\mathrm{rbst\_sub}} = 1- \sum_{n=1}^{N}{N\choose n}(\gamma_{\mathrm{sub}})^n\Big(1-(1-p/L)^{KT}\Big)^{n}(1-p/L)^{KT(N-n)}$}
\resizebox{0.85\hsize}{!}{$= 1+ \Big(1-(1-p/L)^{KT}\Big)^N - \Big(1-(1-p/L)^{KT}+\gamma_{\mathrm{sub}}(1-p/L)^{KT}\Big)^N$}.  
Clearly, the larger the value of $p$, the less the difference between 
\resizebox{1\hsize}{!}{$1-(1-p/L)^{KT}$ and  $1-(1-p/L)^{KT}+\gamma_{\mathrm{sub}}(1-p/L)^{KT}$, which} 
 suggests that $P_{\mathrm{rbst\_sub}}$ also monotonically increases with $p$. 


\subsubsection{Robustness Against   Correlation Attack}\label{sec:robs_corr_atk}

In \cite{ji2021curse}, the authors identify a correlation attack against database fingerprinting mechanisms, which takes  advantage of the intrinsic correlation between data entries in the database to infer and compromise  the potentially fingerprinted bit positions. In particular, the  malicious SP changes the insignificant bits of entries in $\widetilde{\R}$ if the data entries satisfy \resizebox{1\hsize}{!}{$\Big|\Pr\Big(  \widetilde{ \R[t]}=\pi, \widetilde{\R[z]} = \omega  \Big)-\Pr\Big(\R[t] = \pi,\R[z] = \omega\Big)\Big|\geq \tau,\forall z\in[1,T],\forall \omega$}, where $\tau$ is a predetermined parameter for this attack.

Similar to \cite{ji2021curse}, we adopt the confidence gain of the malicious SP (denoted as $G$)   to analyze the robustness of the proposed fingerprinting mechanism against the correlation attack. The confidence gain measures the   knowledge of a potentially fingerprinted data entry under correlation attack over random guess.  To be more specific, $G$ is defined as the ratio between the probability that a specific entry (whose original $t$th attribute  takes value $\pi$) will be selected to be compromised in the correlation attack and the probability that such entry will be selected to be compromised in the random bit flipping attack. Mathematically, this can be shown as \resizebox{1\hsize}{!}{$G = \frac{1-\prod_{z\in[1,t],z\neq t}\prod_{\omega}\Pr\Big(\Big|\Pr\big(  \widetilde{ \R[t]}=\pi, \widetilde{\R[z]} = \omega  \big)-\Pr\big(\R[t] = \pi,\R[z] = \omega\big)\Big|\leq \tau\Big)}{(1-(1-p)^K)\Pr(\R[t]=\pi)}$}.

In \cite{ji2021curse}, the authors have   shown that $G$ decreases as the percentage of fingerprinted entries increases (when considering the fingerprinting mechanism developed in \cite{li2005fingerprinting}). In Appendix \ref{sec:G_vs_p}, we present the similar analysis and show that $G$ also decreases as $p$ increases when our proposed entry-level differentially-private fingerprinting mechanism is used. As a result, this implies that the robustness of our proposed fingerprinting mechanism also increases with $p$. In Section \ref{sec:discussion}, we will discuss how to augment our proposed   mechanism to mitigate the correlation attacks.

%% file: sections/multiple_query_new_version.tex
\section{Sharing Multiple Fingerprinted Databases}\label{sec:multiple_sharing}


A major challenge in practical use of differential privacy is that data privacy  degrades if the same statistics are repeatedly calculated and released using the same differentially-private mechanism. The same is true for sharing   a database with multiple SPs. 
If  different fingerprinted copies of the same database are shared multiple times , the average of them may   converge to the original database, which implies that    privacy guarantee of Algorithm \ref{algo:fingerprint_sp_n} degrades linearly with   number of sharings, as it is used to share the same database repeatedly. 

However, in practice, the database owner releases its database only to a limited number of SPs, and for each released copy, it will have certain data privacy and fingerprint robustness  requirements.  According to Figure \ref{fig:relationships},  these   requirements can both be fulfilled 
if the utility of the shared database is compromised to a certain extent (but not significantly as will be corroborated in Section \ref{sec:experiment}). Based on Section \ref{sec:connections}, we know that the database utility   is readily be  characterized by the fingerprint density, as defined   in Corollary \ref{corollary:density}, because the higher the fingerprint density, the lower the database utility, and the database owner can control the utility of repeatedly shared databases using fingerprint density. As a result, we let   the database owner   only share a fingerprinted database, $\mathcal{M}(\R)$, if its fingerprint density, $||\mathcal{M}(\R)-\R||_{1,1}$, is beyond a predetermined publicly known threshold, $\Gamma$, in order to meet the database owner's requirement of high data privacy guarantee and fingerprint robustness. 

As discussed in Section \ref{sec:an_instant}, the fingerprinted database $\mathcal{M}(\R)$ customized for a particular SP depends on an internal ID assigned by the database owner to the corresponding SP. Since the internal ID of the SP is an input for inserting the  fingerprint (Algorithm \ref{algo:fingerprint_sp_n}), whether $||\mathcal{M}(\R)-\R||_{1,1}$ is higher than $\Gamma$ also depends on the assigned internal ID. 
As a consequence, when an SP queries the database, the database owner needs to keep generating a new internal ID for it until the resulting fingerprint density is above $\Gamma$. 
Moreover, this process (i.e., internal ID generation and fingerprint density comparison with the threshold) also needs to be performed in a privacy-preserving manner. The reason is that according to Section \ref{sec:connections} (specifically, Proposition \ref{thm:cell_accuracy} and Corollary \ref{corollary:density}),  fingerprint density    provides  additional knowledge about the   fingerprint robustness and differential privacy guarantee. If a malicious SP accurately knows that its received database has fingerprint density  higher than a   threshold $\Gamma$, it can estimate the percentage of changed entries due to fingerprint, and it can further distort the fingerprint via a correlation attack \cite{ji2021curse}.  



The above discussion inspires us to resort to the sparse vector technique (SVT) \cite{dwork2014algorithmic,lyu2017understanding}, that only releases a noisy  query result when it is beyond a noisy threshold, to design a mechanism for sharing   multiple entry-level differentially-private fingerprinted databases and at the same time controlling the cumulative privacy loss. 
The unique benefit of SVT is that it can answer multiple queries  while paying the cost of privacy only for the ones satisfying a certain condition, e.g.,   when the result is beyond a given threshold. 
In Section \ref{sec:inter_step}, we present an intermediate step which considers only one SP,  determines its internal ID, and conducts the comparison between the resulting fingerprint density and threshold  under differential privacy guarantee. In Section \ref{sec:multiple_shares}, we show how to compose this intermediate step for $C$ times to determine the internal IDs for $C$ SPs and share different fingerprinted databases with them.





\subsection{Intermediate Step: Determining   Internal ID for One SP}\label{sec:inter_step}


As elaborated earlier, the database owner needs to assign an internal ID to a SP in order to achieve  $||\mathcal{M}(\R)-\R||_{1,1}>\Gamma$ for the purpose of simultaneously meeting  data privacy and fingerprint robustness requirements. To achieve differential privacy for this intermediate step, we perturb both $||\mathcal{M}(\R)-\R||_{1,1}$ and $\Gamma$, and consider the noisy comparison $||\mathcal{M}(\R)-\R||_{1,1}+\mu>\Gamma+\rho$, where $\mu$ and $\rho$ are Laplace noises. 
Establishing  the noisy comparison is a standard approach in SVT (see \cite{dwork2014algorithmic} page 57, and \cite{lyu2017understanding} page 639).

Next, we formally present the intermediate step. 
When the database owner receives a query from a new SP (suppose that this SP is the $c$th SP and $c\in[1,C]$), it  generates an instance of internal ID for the $c$th SP via $ID_{\mathrm{internal}}^c = Hash(\mathcal{K}|c|i)$, where $i\in\{1,2,3\cdots\}$ denotes the sequence number of this trial to generate $ID_{\mathrm{internal}}^c$. Then, the database owner generates a copy of fingerprinted database by calling Algorithm \ref{algo:fingerprint_sp_n} with the internal ID set as $ID_{\mathrm{internal}}^c$ in Line \ref{line:set_ID}. Similarly, we denote the fingerprinted database generated for the $c$th SP at the $i$th trial as $\mathcal{M}_i^c(\R)$. Next, the database owner conducts the noisy comparison $||\mathcal{M}_i^c(\R)-\R||_{1,1}+\mu_i>\Gamma+\rho_i$, 
where $\mu_i\sim\mathrm{Lap}(\frac{\Delta}{\epsilon_2})$ and $\rho_i\sim\mathrm{Lap}(\frac{\Delta}{\epsilon_3})$.  Here, $\epsilon_2$ and $\epsilon_3$ are the privacy budgets used to control the accuracy of the noisy comparison. If $||\mathcal{M}_i^c(\R)-\R||_{1,1}+\mu_i>\Gamma+\rho_i$ holds, then the database owner returns a symbol $\top$ and immediately terminates the intermediate step. This means that  $ID_{\mathrm{internal}}^c$ generated at the $i$th trial for the $c$th SP can lead to a fingerprinted database satisfying the data privacy and fingerprint robustness requirements. Otherwise, the database owner returns a symbol $\bot$, increases $i$ by 1, and continues the process. We summarize this intermediate step in Algorithm \ref{algo:DeterminingQualifiedSP}.  
This entire process is differentially-private, and is proved in the following theorem. It is noteworthy that the identified $ID_{\mathrm{internal}}^c$ is not released to the SP. As we will show in Section \ref{sec:exp_multiple_shares}, an instance of $ID_{\mathrm{internal}}^c$ that leads to $||\mathcal{M}_i^c(\R)-\R||_{1,1}+\mu_i>\Gamma+\rho_i$ can usually be generated in 1 or 2 trials depending on the ratio between $\epsilon_2$ and $\epsilon_3$.

\begin{algorithm}
\small
\SetKwInOut{Input}{Input}
\SetKwInOut{Output}{Output}
\caption{Determine the Internal ID for One SP}\label{algo:DeterminingQualifiedSP}
\Input{Original database $\R$, fingerprinting scheme $\mathcal{M}$,   sequence number of a new SP, i.e., $c$,   threshold $\Gamma$, and privacy budget $\epsilon$, $\epsilon_2$,  and $\epsilon_3$.} 
\Output{$\{\bot,\bot,\cdots,\bot,\top\}$.}

\ForAll{$i\in\{1,2,3,\cdots\}$}{

Generate   an instance of internal ID for the $c$th SP via $ID_{\mathrm{internal}}^c = Hash(\mathcal{K}|c|i)$.

Generate   $\mathcal{M}_i^c(\R)$ by calling  Algorithm \ref{algo:fingerprint_sp_n} with $ID_{\mathrm{internal}}^c$ and 
privacy budget $\epsilon$.\label{line:generate_MR}

Sample $\mu_i\sim\mathrm{Lap}(\frac{\Delta}{\epsilon_2})$ and $\rho_i\sim\mathrm{Lap}(\frac{\Delta}{\epsilon_3})$.


\eIf{$||\mathcal{M}_i^c(\R)-\R||_{1,1} +\mu_i \geq \Gamma+ \rho_i$\label{line:generate_rho}}{Output $a_i = \top$.\algorithmiccomment{\textbf{the $i$th trial meets the requirements}}

{Terminate the algorithm}.
}{Output $a_i = \bot$.\algorithmiccomment{\textbf{the $i$th trial does not meet the requirements}}}
}
\end{algorithm}

\begin{theorem}\label{thm:proof_qualified}
Algorithm \ref{algo:DeterminingQualifiedSP} is ($\epsilon_2+\epsilon_3$)-entry-level differentially private.
\end{theorem}

We show the proof of Theorem \ref{thm:proof_qualified} in Appendix \ref{sec:proof_qualified_section}. Note that although we allocate the privacy budget $\epsilon$ in the course of  generating the fingerprinted database for the SP at Line \ref{line:generate_MR} in Algorithm \ref{algo:DeterminingQualifiedSP}, it does not contribute to the total privacy loss. This is because here, we  only determine $ID_{\mathrm{internal}}^c$ that is used for fingerprint insertion, but the numerical fingerprinted database has not been shared yet.


\subsection{Composition of Intermediate Steps: Releasing Multiple Numerical Fingerprinted Databases}\label{sec:multiple_shares}


In the previous section, we have presented an intermediate step, in which, to guarantee that an SP receives a copy of fingerprinted database, the database owner keeps generating an instance of internal ID for it until the noisy comparison result is ``TRUE''. 
Here, we compose the intermediate steps for $C$ times to determine the internal IDs for $C$ SPs, and at the same time, share the corresponding fingerprinted databases (generated using their final internal IDs) with them. 

We first restate the advanced composition theorem, which provides a tight cumulative privacy loss for adaptive composition of  differentially-private mechanisms (see \cite{dwork2014algorithmic} page 49), in context of our adopted privacy model, i.e.,  entry-level differentially privacy.

\begin{theorem}[Advanced Composition]\label{thm:comp}
For all $\epsilon,\delta,\delta'\geq 0$, the $C$-fold   composition of $(\epsilon,\delta)$-entry-level differentially private mechanisms satisfy \resizebox{0.65\hsize}{!}{$\Big(\sqrt{2C\ln(1/\delta')}\epsilon+C\epsilon(e^\epsilon-1),C\delta+\delta'\Big)$}-entry-level differential privacy.
\end{theorem}

Theorem \ref{thm:comp} can be proved by following the exact procedures of proving the advanced composition theorem for conventional differential privacy with a simple switch of the privacy model.  In Theorem \ref{thm:comp}, $\delta'$ is determined by the database owner to ensure the desired $\delta'$-approximate max divergence between two random variables (\cite{dwork2014algorithmic} page 43). Using this, we propose Algorithm \ref{algo:MultipleFingerprintedDatabasesReleasing} (shown in Appendix \ref{sec:appendix_multiple_shares}),  which outputs the numeric fingerprinted databases to  $C$ SPs, and all shared database copies meet the database owner's requirements for data privacy and fingerprint robustness. 
The workflow of Algorithm \ref{algo:MultipleFingerprintedDatabasesReleasing} is similar to that of Algorithm \ref{algo:DeterminingQualifiedSP}, the only difference is that Algorithm \ref{algo:MultipleFingerprintedDatabasesReleasing} outputs the numerical fingerprinted database generated by the final internal ID of each SP instead of the symbol $\top$. We include a box around Line \ref{line:output_fpdp} in Algorithm \ref{algo:MultipleFingerprintedDatabasesReleasing} to highlight this.  
Whereas, Algorithm \ref{algo:DeterminingQualifiedSP} just identifies one proper internal ID for each SP. Note that if the database owner wants to share its database with more than $C$ different SPs, it can reduce the value of the fingerprint density threshold $\Gamma$, which, however, compromises the privacy and fingerprint robustness of the shared databases, because reducing $\Gamma$ increases utility of the shared database.

\begin{theorem}\label{thm:multiple_share_is_dp}
Algorithm \ref{algo:MultipleFingerprintedDatabasesReleasing} is ($\epsilon_0,\delta_0$)-entry-level differentially private (defined in Definition \ref{def_dp}) with  $\epsilon_0 = \sqrt{2C\ln(\frac{1}{\delta'})}(\epsilon+\epsilon_2+\epsilon_3)+C\Big(\epsilon(e^{\epsilon}-1)+(\epsilon_2+\epsilon_3)(e^{\epsilon_2+\epsilon_3}-1)\Big)$, and $\delta_0 = 2\delta'$.
\end{theorem}
We show the proof of Theorem \ref{thm:multiple_share_is_dp} in Appendix \ref{sec:multiple_share_is_dp_proof}.

\noindent\textbf{Privacy budget allocation.} In practice, given the cumulative privacy budget $\epsilon_0$ and $\delta'$ specified by the database owner, we need to decide how to assign the values of $\epsilon$, $\epsilon_2$, and $\epsilon_3$. Since $\epsilon$ is used to obtain the fingerprinted database, its value should be determined based on the specific database of interest. Specifically, as we have shown in Figure \ref{fig:relationships}, a lower $\epsilon$ leads to lower database utility and a higher $p$ value (the probability of changing one insignificant bit of a data entry), which, in turn, increases fingerprint robustness. As a result, the database owner can decide $\epsilon$ based on its requirements about database utility and fingerprint robustness. 

Furthermore, we note that $(\epsilon_2+\epsilon_3)$ is used to obtain the internal IDs of SPs. Once $\epsilon$ is decided, the database owner can solve for $(\epsilon_2+\epsilon_3)$ numerically, i.e.,
\begin{equation}\label{eq:numerical_eq}
\resizebox{1\hsize}{!}{$\Big(\sqrt{2C\ln(\frac{1}{\delta'})}-C\Big) (\epsilon_2+\epsilon_3)+ C(\epsilon_2+\epsilon_3)e^{\epsilon_2+\epsilon_3} = \epsilon_0-\Big(\sqrt{2C\ln(\frac{1}{\delta'})}-C\Big)\epsilon-C\epsilon e^{\epsilon}$}
\end{equation}

Suppose the numerical solution is $(\epsilon_2+\epsilon_3) = \epsilon^*$. Then, we need to allocate $\epsilon^*$ to $\epsilon_2$ and $\epsilon_3$. We observe that $\epsilon_2$ and $\epsilon_3$ control the accuracy of noisy comparison, i.e., the comparison between the perturbed fingerprint density and the perturbed threshold, $||\mathcal{M}_i(\R)-\R||_{1,1} +\mu_i \geq \Gamma+\rho_i$ (or equivalently,   $||\mathcal{M}_i(\R)-\R||_{1,1} -\Gamma\geq \rho_i-\mu_i$) in both Algorithms \ref{algo:DeterminingQualifiedSP} and \ref{algo:MultipleFingerprintedDatabasesReleasing}. 
To boost the accuracy of the noisy comparison, we minimize the variance of the difference between $\rho_i$ and $\mu_i$. Since they are both Laplace random variables, the variance of their difference is $2(\frac{\Delta}{\epsilon_2})^2+2(\frac{\Delta}{\epsilon_3})^2$. Clearly, given $\epsilon^*$, the variance is minimized when $\epsilon_2=\epsilon_3=\epsilon^*/2$.  Note that in classical SVT, the database owner does not respond to all the queries, i.e., it merely reports ``$\bot$'' if the considered noisy comparison is ``FALSE'' (see \cite{dwork2014algorithmic} page 55), however, this is not user-friendly in database sharing, especially, when the database owner still has more privacy budget remained. In our proposed SVT-based solution, we make sure all SPs get their fingerprinted databases as long as they are among the top $C$ SPs sending the query request. This is achieved by letting the database owner keep generating new internal IDs for the SPs, until the noisy comparison turns out to be ``TRUE'', and this approach does not violate the design principle of SVT.

%% file: sections/experiment.tex
\section{Experiments}
\label{sec:experiment}

Here, we evaluate the developed differentially-private relational database fingerprinting mechanism under both single and multiple database sharing scenarios.  Note that some experiment results are deferred to Appendix \ref{sec:miss_exp}.

\subsection{Experiment Setup}\label{sec:setup}
The database we used for the evaluations is the nursery school application database \cite{asuncion2007uci}, which contains data records of  12960 applicants. Each applicant has 8 categorical attributes, e.g., ``form of the family'' (complete, completed, incomplete, or foster),  and ``number of children in the family'' (1, 2, 3, or more). 
Each data record belongs to one of the five classes, i.e.,    ``not\_recom'',  ``recommend'',      ``very\_recom'',   ``priority'', and     ``spec\_prior'', indicting the decisions for the applicants.

\noindent\textbf{Database encoding.} To fingerprint the categorical data, we first represent all possible values of each attribute as an integer starting from 0.  In the considered database, the maximum integer representation of a data entry is 4 (we do not fingerprint the labels, which will be used in a classification task to evaluate the utility of the fingerprinted database). We also perform zero padding to make sure that the binary representations of different data entries  have the same length. 
Note that our proposed mechanism is not limited to discrete or categorical attributes only. If a relation has continuous attributes, we can first quantize the entry  values into non-overlapping ranges, and then represent each range interval using an integer \cite{ji2021curse}.

\noindent\textbf{Sensitivity Control.} Since the integer representations of data entries vary from 0 to 4,  the sensitivity is $\Delta=4$. Thus, the proposed mechanism needs to fingerprint $K = \log_24+1 =3$ (see Theorem \ref{thm:dpfp}) least significant bits of each data entry. This, however, may significantly compromise the utility of the fingerprinted database.  To control the sensitivity (and hence improve the utility), we make the following observation.  
We calculate the fraction of pairwise absolute differences taking a specific value  (between the attributes) and show the results in Table \ref{table:pairwise_diff} in Appendix \ref{sec:fraction_diff_result}.  Clearly, in each class, a large portion of the absolute differences are 0 and 1, and only a small fraction of them have difference larger than 1. 
Thus, in the  experiments, we consider sensitivity $\Delta=1$ with the assumption that the different entries in a pair of neighboring nursery databases can change by at most by 1, otherwise, it introduces a rare event (i.e., data record that occurs with very low probability) in the database. Our approach to control the sensitivity is similar to the restricted sensitivity   \cite{blocki2013differentially} (which  calculates sensitivity on a restricted subset of the database, instead of quantifying over all possible data records) and smooth sensitivity \cite{kellaris2013practical} (which smooths the data records after partitioning them into non-overlapping groups).  
Note that it has been widely recognized that rare events or outliers consume extra privacy budget, and this is a common problem in differentially-private database queries  \cite{chaudhuri2006random,lui2015outlier,dwork2009differential,li2012sampling}. Controlling local and global sensitivity in differential privacy is a separate topic, and it is beyond the scope of this paper. 



\noindent\textbf{Post-processing.} After fingerprinting a chosen database ($\R$),  some entries may have integer representations that are outside the domain of the considered database.  
Thus, we also conduct a post-processing on the resulting database ($\mathcal{M}(\R)$) to eliminate the   data entries that are not in the original domain. Otherwise, the SP which receives the database can understand that these entries are changed due to fingerprinting. 




\subsection{Evaluations for One-time Sharing}

We first consider the scenario, where the database owner only releases the fingerprinted database to one SP. Thus, only Algorithm \ref{algo:fingerprint_sp_n} is invoked.

\subsubsection{Fingerprint Robustness}  

Among  common attacks against database fingerprinting mechanisms (i.e., random flipping attack, subset attack, and superset attack~\cite{li2005fingerprinting}), random flipping attack is shown to be the most powerful one \cite{li2005fingerprinting,yilmaz2020collusion},  thus,  we investigate the  robustness of the proposed mechanism against this attack.  In Section \ref{sec:discussion}, we will discuss how to make the proposed mechanism robust against more sophisticated attacks, e.g., the collusion attack \cite{yilmaz2020collusion} and correlation attack \cite{ji2021curse}. In particular, in favor of the malicious SP (and in order to show the robustness of our mechanism), we let the malicious SP randomly flip 80\% of the bit positions in its received copy of fingerprinted database $\mathcal{M}(\R)$, i.e., $\overline{\R}$ (the compromised database) is obtained by flipping $80\%$ bits in $\mathcal{M}(\R)$. Then, we measure the fingerprint robustness using  the number of bit matches between the malicious SP's fingerprint   and the one extracted from $\overline{\R}$. 

We  compare the fingerprint robustness with a state-of-the-art database fingerprinting mechanism proposed in \cite{li2005fingerprinting}.  The reason we choose this mechanism  is because  it also has high robustness, e.g., the probability of detecting no valid fingerprint due to  random bit flipping attack   is upper bounded by ${(|SP|-1)}/{2^L}$ (here $|SP|$ is the number of SPs who have received the fingerprinted copies).  This mechanism controls the fingerprint density using fingerprint fraction (i.e., the fraction of fingerprinted data entries in $\R$) denoted as $\lambda$. 



\begin{wraptable}{r}{.7\linewidth}
    \begin{minipage}{1\linewidth}
\begin{adjustbox}{width=1\textwidth}
\begin{tabular}{| *{5}{c|} }
    \hline
 \multirow{2}{*}{$\epsilon$}&    $p$ set   & $\lambda$ set  &  \cellcolor{Red}bit matches  &  bit matches\\
& by us & in \cite{li2005fingerprinting} & \cellcolor{Red}obtd. by us & obtd. by \cite{li2005fingerprinting}\\
    \hline
1 &0.2689   &  21.38\%  &\cellcolor{Red}128 & 128\\  
2 &0.1192    & 9.75\% &\cellcolor{Red}127   &122\\ 
3 &0.0474    & 3.90\%&\cellcolor{Red}120  &96 \\ 
4 &0.0180    & 1.64\%&\cellcolor{Red}106   &82\\ 
5 &0.0067    & 0.65\%&\cellcolor{Red}84  &72\\ 
6 &0.0025    & 0.24\%&\cellcolor{Red}71  &50\\ 
7 &0.0009    & 0.07\%&\cellcolor{Red}67   &30\\ 
    \hline
\end{tabular}
  \end{adjustbox}
   \end{minipage}
   \captionsetup{font=footnotesize}
\caption{Comparison of fingerprint robustness achieved by us and \cite{li2005fingerprinting} when  random flipping attack  changes 80\%    bit positions in $\mathcal{M}(\R)$. Each row in the table corresponds to the same database utility for both mechanisms.} 
\label{table:dpfp_vs_vannila}
\end{wraptable} 
We present the comparison results in Table \ref{table:dpfp_vs_vannila}. In particular, we vary the privacy budget of our entry-level  differentially-private fingerprinting mechanism, i.e., $\epsilon$ in Algorithm \ref{algo:fingerprint_sp_n}, from 1 to 7, set the corresponding probability of changing an insignificant bit as $p = 1/(e^{{\epsilon}/{K}}+1)$ (according to Theorem \ref{thm:dpfp}), and  set the fingerprint fraction ($\lambda$) used in \cite{li2005fingerprinting} to be the percentage of changed data entries caused by our proposed mechanism under varying $\epsilon$ followed by post-processing. 
Clearly, our proposed mechanism (the highlighted column in Table \ref{table:dpfp_vs_vannila}) achieves higher fingerprint robustness. Note that the database utility of the two mechanisms are identical in  each row in Table \ref{table:dpfp_vs_vannila}. As $p$ (or $\lambda$) decreases (i.e., less bit positions are marked by either fingerprint mechanisms), the number of matched fingerprint bits between the extracted fingerprint from $\overline{\R}$ and the malicious SP's fingerprint also decreases. However, the robustness of the proposed scheme remains significantly high compared to~\cite{li2005fingerprinting}. 
Especially, when $p\leq 0.0025$, even if the malicious SP changes $80\%$ of the bits in $\mathcal{M}(\R)$, it is still not able to compromise more than half of the fingerprint bits (i.e., 64 out of 128) generated by the proposed mechanism. It has been empirically shown that as long as the malicious SP does not distort  more than half of the fingerprint bits, it will end up being uniquely identifiable (by the fingerprint detection algorithm of the database owner)~\cite{li2005fingerprinting}. In contrast, when $\lambda\leq 0.24\%$, the malicious SP can compromise more than half of the fingerprint bits generated by \cite{li2005fingerprinting}, which means that it will not be accused although it has illegally redistributed the data.  In addition to higher fingerprint robustness, our proposed mechanism also achieves   privacy guarantee at the same time.


In Section \ref{sec:robust_rand_flip}, we derive the 
 robustness (denoted as $P_{\mathrm{rbst\_rnd}}$) of our mechanism against random bit flipping attack in terms of flipping probability $\gamma_{\mathrm{rnd}}$. In    Appendix \ref{sec:analysis_rnd}, we show that $P_{\mathrm{rbst\_rnd}}$ is monotonically increasing with $p$ ($p\in(0,0.5)$), i.e., we showed how the robustness of the proposed mechanism against   random bit flipping attack increases with the increasing number data entries are changed due to fingerprinting. 
 By letting $D=64$ (threshold of number of matched fingerprint  bits, see   Section \ref{sec:robust_rand_flip}), $\gamma_{\mathrm{rnd}}=0.8$, and setting $p$ as the values reported in Table \ref{table:dpfp_vs_vannila}, we plot the numerical value of $P_{\mathrm{rbst\_rnd}}$  in Figure \ref{fig:robust_rnd} (blue line) in Appendix \ref{sec:app_robust_rnd}. 
 In the same figure, the green line is the percentage  of bit matches obtained by our mechanism as a result of our empirical evaluation (i.e., values in the highlighted columns in Table \ref{table:dpfp_vs_vannila} divided by 128). We observe that both lines in the figure almost overlap, which validates our theoretical  findings. Besides, we have also evaluated the  averaged error  of   entries in the fingerprinted database, and we observe that empirical value is close to $\Delta p$ as proved in Proposition \ref{thm:cell_accuracy}.
 
In addition to fingerprint robustness, in Appendix \ref{sec:app_att_inf_atk}, we also evaluate the privacy guarantee of the proposed mechanism by empirically investigating the inference capability of an adversary that is formulated in \cite{liu2016dependence}. Our experimental results show that attacker's inference capability is always bounded by our theoretical findings in Proposition \ref{prop:att_inf_bound} under varying $\epsilon$, which suggests the robustness of our mechanism against attribute inference attack.


\subsubsection{Utility of the Shared Database}\label{sec:two-stage-approach}

Here, to experimentally show the utility guarantees of the proposed entry-level  differentially-private fingerprint mechanism, we compare the utility of it with that achieved by a two-stage approach, i.e., first perturbing the entire database under  local  differential privacy with the identical $\epsilon$ and then fingerprinting the perturbed database using \cite{li2005fingerprinting} as before. 
We conduct the comparison from two perspectives: (i) application task-independent comparison, which considers the change of variance of each attribute caused  by  our   mechanism and the two-stage approach and the accuracy of specific SQL queries after fingerprinting, and (ii) application task-specific  comparison, where we use   fingerprinted databases (ours and the two-stage approach) to do   classification   and principal component analysis (PCA).

\noindent\textbf{Application task-independent comparison.} Please refer to Appendix \ref{sec:exp_task_indpend} for the setup and detailed results of these experiments. Here, we only highlight the empirical findings. Our proposed mechanism always achieve higher utility, i.e., the  changes in invariance of each attribute caused  by  our  mechanism  are  generally  10  times smaller  than  that  caused  by  the  two-stage approach, and our approach also achieves higher SQL query accuracy given different choices of $\epsilon$. These imply that our proposed mechanism can also achieve higher utility for various task-specific applications, because in data mining or machine learning, the results usually depend on the intrinsic data distributions, and more accurate variance (or covariance matrices) lead to better performance. We will validate this experimentally  in the following.

\noindent\textbf{Task-specific comparison.} As discussed in Section \ref{sec:setup}, data records belong to 5 classes. To perform   classification, we adopt a multi-class support vector machine (SVM) classifier  and use 65\% of data records for training and the rest for testing. We evaluate the utility of various fingerprinted databases by comparing the \textbf{fingerprinted testing accuracy} (i.e., SVM classifier trained on fingerprinted training data and then tested on the original testing data) with the \textbf{original testing accuracy} (i.e., SVM classifier trained on the original training data and then tested on the original testing data). Thus, the less the difference between fingerprinted testing accuracy and original testing accuracy, the higher the database utility.

The database utility for PCA is defined as the total deviation,   $\mathrm{TTL_{DEV}} =   \textstyle\sum_{i=1}^{8} |\lambda_i-\widetilde{\mathbf{v}_i}^T  \mathbf{C}  \widetilde{\mathbf{v}_i}|$, where   $\mathbf{C}$ is the  empirical  covariance matrix obtained from the original (non-fingerprinted) database, $\lambda_i$ values are the eigenvalues of $\mathbf{C}$, and $\widetilde{\mathbf{v}_i}$ vectors are the eigenvectors of the empirical  covariance matrix obtained from    fingerprinted database.   Specifically,  $\mathrm{TTL_{DEV}}$ quantifies the deviation of the variance (of the fingerprinted database) from $\lambda_i$ in the direction of the $i$th component of $\mathbf{C}$. The smaller the deviation ($\mathrm{TTL_{DEV}}$) is, the higher the utility.

In Figure \ref{fig:task-specific}(a) and (b) (see Appendix \ref{sec:app_exp_task_specific}) by varying $\epsilon$ from 0.25 to 2, we compare the task-specific database utilities for   classification and PCA achieved by our mechanism and the two-stage approach. Clearly, our proposed mechanism achieves higher database utilities in both considered applications. As future work, we will also investigate database utility in more sophisticated tasks, such deep learning and federated learning, where stochastic gradient descent will be iteratively calculated using the fingerprinted databases.

\subsection{Evaluations for Multiple Sharing}\label{sec:exp_multiple_shares}
Next, we consider the scenario where at most 100 SPs query the entire database over time in a sequential order (i.e., $C=100$  in order to control the cumulative privacy loss). As discussed in Section \ref{sec:multiple_sharing}, the database owner performs the noisy comparison ($||\mathcal{M}_i^c(\R)-\R||_{1,1}+\mu_i\geq \Gamma +\rho_i$, where $\Gamma$ is the fingerprint density, $\mu_i$ and $\rho_i$ are Laplace noises) to   determine the proper internal IDs for the SPs  to generate the  fingerprinted databases. In the experiment, we set $\Gamma = (\frac{1}{2}+\frac{1}{\sqrt{12}})\Delta p NK$. The rationality  is that according to Corollary \ref{corollary:density}, the expected value of $||\mathcal{M}_i^c(\R)-\R||_{1,1}$ falls in the range of $[0,\Delta pNT]$.  Since we do not have any assumption on the database, and the pseudorandom number generator $\mathcal{U}$ generates each random number with equal probability, we approximately model   $||\mathcal{M}_i^c(\R)-\R||_{1,1}$  as  a uniformly distributed  random variable   in the range of $[0,\Delta p NT]$. Then, its  mean  and standard deviation are  $\frac{1}{2}\Delta p NT$ and   $\frac{1}{\sqrt{12}}\Delta p NK$. 


Moreover, we consider the cumulative privacy loss as $\epsilon_0 = 40$ and $\delta_0 = 2*10^{-3}$.  If $\epsilon_0<40$ and the database owner   still wants to generate fingerprinted databases with the identical privacy and fingerprint robustness guarantees as when $\epsilon_0=40$, it will end up sharing its database with fewer number of SPs. To achieve a decent  database utility, 
we set the privacy budget to generate the entry-level differentially-private fingerprinted database as $\epsilon = 0.5$. Then, by solving (\ref{eq:numerical_eq}) numerically, we have $\epsilon_2+\epsilon_3=0.002$ approximately.

\begin{wraptable}{r}{.75\linewidth}
    \begin{minipage}{1\linewidth}
\begin{adjustbox}{width=1\textwidth}
\begin{tabular}{| *{6}{c|} }
    \hline
$\epsilon_2:\epsilon_3$&    $9:1$    & $7:1$   & $5:1$  & $3:1$  &  $1:1$\\
    \hline
Total No. of trials& 181 & 177 & 173 & 165 & 156 \\
    \hline
\end{tabular}
  \end{adjustbox}
   \end{minipage}
    \captionsetup{font=footnotesize}
\caption{Impact of the ratio between $\epsilon_2$ and $\epsilon_3$ on the total number of internal ID generation trials for 100 SPs.}
\label{table:ratio_impact}
\end{wraptable} 
We first investigate the impact of privacy allocation between  $\epsilon_2$  and $\epsilon_3$  on the total number of trials to determine the proper internal IDs for all 100 SPs. In particular, we vary the ratio between $\epsilon_2$ and $\epsilon_3$ from  $9:1$ to $1:1$, and show the results in Table \ref{table:ratio_impact}. We observe that as the difference between $\epsilon_2$ and $\epsilon_3$ decreases (their ratio decreases), Algorithm \ref{algo:MultipleFingerprintedDatabasesReleasing} terminates with fewer internal ID generation trials. Especially, when $\epsilon_2:\epsilon_3 = 9:1$, 81 (181 instead of 100) additional trials are made, whereas, when $\epsilon_2 : \epsilon_3 = 1:1$, only 56 (156 instead of 100) additional  trials are made. 
Since the cumulative privacy loss is identical for the different total  number of trials reported in Table~\ref{table:ratio_impact}, this suggests that when $\epsilon_2 = \epsilon_3$, the internal ID generation efficiency is higher from the perspective of the data recipients (i.e., the probability of the database owner can generate a proper internal ID for an SP in 1 or 2 trials increases). This finding validates our suggestion of equally dividing  the privacy budget between $\epsilon_2$ and $\epsilon_3$  to reduce the number of trials to generate proper internal IDs. Besides, we also would like to highlight that by adopting SVT, we  significantly reduce the cumulative privacy loss, because otherwise, the privacy loss will be $(\epsilon+\epsilon_2+\epsilon_3)\times (\mathrm{Total\ No.\ of\ trials})$. For instance, when  $\mathrm{Total\ No.\ of\ trials} = 181$, the privacy would be $\epsilon = (0.5+0.02)\times 181 = 90.862$ without SVT (compared to a cumulative privacy loss of $40$ with SVT).


Since the fingerprinted databases shared with different SPs   are all obtained under $\epsilon = 0.5$ (but with different fingerprints,  decided by their final internal  IDs), their privacy guarantees and fingerprint robustness   are   identical. In Figure \ref{fig:density_plot} (shown in Appendix \ref{sec:app_multiple_sahres}), we plot the fingerprint densities of the shared databases when $\epsilon_2 = \epsilon_3$, i.e., $||\mathcal{M}(\R)-\R||_{1,1}$ as a result of sharing with 100 SPs. In particular, the red line represents the result achieved by us and the blue line represents the result achieved by the two-stage approach (discussed in Section \ref{sec:two-stage-approach}), which is executed for 100 times in order to share the database with 100 SPs. For each execution of the two-stage approach, we let it have the same privacy and fingerprint robustness guarantees as ours. Clearly, our mechanism not only achieves lower fingerprint density, but it also makes the fingerprint density have smaller variance (the variance of results of ours and the two-stage approach are evaluated as 317 and 871, respectively). It means that the   fingerprinted databases generated by us have more stable utilities, and are promising for real-life applications.


%% file: sections/discussion.tex
\section{Discussion}\label{sec:discussion}

Our work is a first step in   uniting differential privacy and   database fingerprinting. We believe that it will draw attention to other challenges and urgent research problems, which we plan to investigate in the future. 
We list some of the potential extensions of this work in the following. 

\subsection{Augmenting the    entry-level differentially-private mechanism against attribute inference attack using data correlation}\label{sec:diss_att_corr}
As discussed in Section \ref{sec:privacy_against_att_inf}, the database owner can further augment the proposed   mechanism to achieve robustness against correlation based attribute inference attack by involving  
the auxiliary information of $\frac{\Pr(\boldsymbol{r}_i[t]=\zeta_1, \boldsymbol{r}_i[k])}{\Pr(\boldsymbol{r}_i[t]=\zeta_2,    \boldsymbol{r}_i[k])}$ in the mechanism design. Such augmentation can be obtained by taking a closer observation at the proof of Theorem \ref{thm:dpfp} (shown in Appendix \ref{sec:proof_thm_dpfp}). In particular, one intermediate step in the proof of Theorem \ref{thm:dpfp} considers deriving the upper bound of $\prod_{k=1}^{K} \frac{\Pr({\mathbf{r}_{i}}[t,k]\oplus B_{i,t,k} = \widetilde{\mathbf{r}_{i}}[t,k])} {\Pr(\mathbf{r}_{i}'[t,k]\oplus B_{i,t,k}' = \widetilde{\mathbf{r}_{i}}[t,k])}$. By using the correlations among attributes and plugging in the Bernoulli parameter $p$ (the probability of changing an insignificantbit in an entry in the database), this specific term can be upper bounded by $\prod_{k=1}^{K} \frac{\Pr({\mathbf{r}_{i}}[t,k]\oplus B_{i,t,k} = \widetilde{\mathbf{r}_{i}}[t,k])\times \sum_{\boldsymbol{r}_i[k]}\max_{\boldsymbol{r}_i[k]}\Pr(\boldsymbol{r}_i[t], \boldsymbol{r}_i[k])} {\Pr(\mathbf{r}_{i}'[t,k]\oplus B_{i,t,k}' = \widetilde{\mathbf{r}_{i}}[t,k])\times \sum_{\boldsymbol{r}_i[k]}\min_{\boldsymbol{r}_i[k]}\Pr(\boldsymbol{r}_i'[t],    \boldsymbol{r}_i[k])}\leq \prod_{k=1}^{K} \frac{(1-p)\times \sum_{\boldsymbol{r}_i[k]}\max_{\boldsymbol{r}_i[k]}\Pr(\boldsymbol{r}_i[t], \boldsymbol{r}_i[k])} {p \times \sum_{\boldsymbol{r}_i[k]}\min_{\boldsymbol{r}_i[k]}\Pr(\boldsymbol{r}_i'[t],    \boldsymbol{r}_i[k])}$, where the inequality can be achieved by following the similar steps presented in the proof of Theorem \ref{thm:dpfp}. Letting the new bound less than $e^{\epsilon}$, and then set the value of $p$ accordingly,  the proposed mechanism can achieve provable privacy guarantees against attribute  inference  attacks  that use  data  correlation. We will also work along this direction in future work.



\subsection{Mitigation of correlation attacks}
Ji et al. \cite{ji2021curse} have developed a mitigation technique to alleviate the correlation attacks against database fingerprinting.  
Their technique modifies a fingerprinted database to make sure that it has similar column- and row-wise joint distributions with the original database. 
Since their technique only changes the non-fingerprinted  data entries and it can be applied as a post-processing step after any  fingerprinting mechanism, it  can also be utilized following our 
mechanism to defend against the correlation attacks. In case of such an integration, our differential privacy guarantee will still hold because of the immunity property of differential privacy for post-processing~\cite{dwork2014algorithmic}.

\subsection{Incorporation with collusion-resistant fingerprinting}
In this work, we do not consider the collusion attacks (where multiple malicious SPs ally together to generate a pirated database from their unique fingerprinted
copies) on the fingerprinted databases. Several works have  proposed collusion-resistant fingerprinting mechanisms in the literature  \cite{boneh1998collusion,boneh1995collusion,yacobi2001improved,pfitzmann1997asymmetric}. To develop a  differentially-private and collusion-resistant fingerprinting mechanism, one potential solution is to replace the fingerprint generation step (i.e., Line \ref{line:generate_fp} of Algorithm \ref{algo:fingerprint_sp_n}) with the Boneh-Shaw (BS) codes \cite{boneh1995collusion} and decide $p$ (the probability of changing one insignificant bit of an entry) based on $\epsilon_1$ and the number of 1's in the BS codeword. We will explore this extension in future work.

\subsection{Improving database utility by utilizing data distribution}
In many real world applications, especially in machine learning and data mining, data records are usually assumed to follow a prior distribution, i.e., they are drawn i.i.d. from an underlying (multivariate) distribution denoted as $d_{\mathrm{prior}}$. In this case, we can  further improve the utility (e.g., accuracy) of the fingerprinted database by formulating a statistical estimation problem \cite{balle2018improving} as $\mathbf{R}' = \arg \min_{\mathbf{R}'}\mathbb{E}_{d_{\mathrm{prior}}}[\mathbb{E}[\mathcal{M}(\mathbf{R})-\mathbf{R}']]$, where $\mathbf{R}'$ is interpreted as the Bayesian denoised fingerprinted database, which is closer to the original database. Note that $\mathbf{R}'$ has the same privacy guarantee as $\mathcal{M}(\mathbf{R})$ because of the post-processing immunity property of differential privacy.



%% file: sections/conclusion.tex
\section{Conclusions}
\label{sec:conclusion}

In this paper, we have proposed a novel mechanism that unites provable privacy and database  fingerprinting for sharing   relational databases. 
To this end, we have first devised a bit-level random response scheme to achieve $\epsilon$-entry-level differential privacy guarantee for the entire database, and then developed a concrete entry-level  differentially-private database fingerprinting mechanism on top of it. 
We have also provided the closed form expressions to characterize the connections between database utility, privacy protection, and fingerprint robustness.  
Finally, we have developed a SVT-based solution to share entry-level differentially-private fingerprinted databases with multiple recipients, and at the same time, control the cumulative privacy loss due to   repeated sharing. 
Experimental results on a real  relational database show that under the same database utility requirement, our mechanism provides stronger fingerprint robustness (along with   privacy guarantee) than other  
mechanisms. Under the same   privacy and fingerprint robustness guarantee,   our   mechanism achieves higher utility than   database perturbation followed by fingerprinting.

%% file: sections/appendix.tex
 \clearpage

\appendices

\section{Frequently Used Notations} \label{sec:notations}
In Table \ref{table:notations}, we list   frequently used notations in the paper.
\begin{table}[htb]
\resizebox{0.8\textwidth}{!}{\begin{minipage}{\textwidth}
\begin{center}
\begin{tabular}{| *{2}{c|} }
    \hline
notations & descriptions\\
\hline
$\R$ & original database\\
\hline
$\R'$ & a neighboring database of $\R$\\
\hline
$\mathcal{M}(\R)$ & fingerprinted database\\
\hline
$\overline{\R}$ & leaked (pirated) database\\
\hline
$\mathbf{r}_i$& $i$th row of $\R$\\
\hline
$\mathbf{r}_i[t,k]$& the $k$th insignificant bit of the $t$th attribute of $\mathbf{r}_i$\\
\hline
 \multirow{2}{*}{$p$} & the probability of changing an insignificant\\
& bit in an entry in the database\\
\hline
 \multirow{2}{*}{$B$} & mark bit to fingerprint a bit position\\
&$B\sim\mathrm{Bernoulli}(p)$\\
\hline
$\epsilon $,$\epsilon_2$,$\epsilon_3$&privacy budgets\\
\hline
$\Delta$ & database sensitivity\\
    \hline
\end{tabular}
\end{center}
\end{minipage}}
    \captionsetup{font=footnotesize}
\caption{Frequently used notations in the paper.} 
\label{table:notations}
\end{table}

\section{Proof of Theorem \ref{thm:dpfp}}\label{sec:proof_thm_dpfp}
\begin{proof}[Proof]  Since we consider  neighboring databases   that have only a pair of different data entries which differ by at most $\Delta$, it requires $K = \floor[\big]{\log_2 \Delta}+1$ bits to encode the difference. Then, by applying Definition \ref{def_dp}, we have
\begin{gather*}
  \begin{aligned}
        & \frac{\Pr\Big(\mathcal{M}(\R) =\widetilde{\R} \Big)}{\Pr\Big(\mathcal{M}(\R') =\widetilde{\R} \Big)}\\
 \stackrel{(a)} = &     \prod_{k=1}^{K} \frac{\Pr\Big({\mathbf{r}_{i}}[t,k]\oplus B_{i,t,k} = \widetilde{\mathbf{r}_{i}}[t,k]\Big)} {\Pr\Big(\mathbf{r}_{i}'[t,k]\oplus B_{i,t,k}' = \widetilde{\mathbf{r}_{i}}[t,k]\Big)}\\
 =    & \prod_{k=1}^{K} \frac{\Pr\Big( B_{i,t,k} ={\mathbf{r}_{i}}[t,k]\oplus \widetilde{\mathbf{r}_{i}}[t,k]\Big)} {\Pr\Big(B_{i,t,k}' = \mathbf{r}_{i}'[t,k]\oplus  \widetilde{\mathbf{r}_{i}}[t,k]\Big)}   \\
 \stackrel{(b)}=  &   \prod_{k=1}^{K}  \frac{p^{\Big({\mathbf{r}_{i}}[t,k]\oplus \widetilde{\mathbf{r}_{i}}[t,k]\Big)}      (1-p)^{\Big(1-{\mathbf{r}_{i}}[t,k]\oplus \widetilde{\mathbf{r}_{i}}[t,k]\Big)} }{p^{\Big(\mathbf{r}_{i}'[t,k]\oplus \widetilde{\mathbf{r}_{i}}[t,k]\Big)}      (1-p)^{\Big(1-\mathbf{r}_{i}'[t,k]\oplus \widetilde{\mathbf{r}_{i}}[t,k]\Big)}}     \\
 \stackrel{(c)} =  &  \prod_{k=1}^{K}  \Big(\frac{1-p}{p}\Big)^{\Big(   ( {\mathbf{r}_{i}}[t,k]- \mathbf{r}_{i}'[t,k] ) (2\widetilde{\mathbf{r}_{i}}[t,k]-1)   \Big)}  \\
\leq  &   \prod_{k=1}^{K}  \Big(\frac{1-p}{p}\Big)^{\Big(   \big| {\mathbf{r}_{i}}[t,k]- \mathbf{r}_{i}'[t,k] \big| (2\widetilde{\mathbf{r}_{i}}[t,k]-1)   \Big)}  \\
\leq   &   \prod_{k=1}^{K}  \frac{1-p}{p},
  \end{aligned}
\end{gather*}
where $(a)$ can be obtained by assuming (without loss of generality) that $\R$ and $\R'$ differ at the $t$th attribute of the $i$th row, and thus the probability ratio at other entries cancel out.  $\mathbf{r}_{i}[t,k]$ (or $\mathbf{r}_{i}'[t,k]$) represents the $k$th least significant bit of the $t$th attribute of $\mathbf{r}_i$ (or $\mathbf{r}_i'$), $B_{i,t,k}$ (or $B_{i,t,k}'$) is the random mark bit fingerprinted on $\mathbf{r}_{i}[t,k]$ (or $\mathbf{r}_{i}'[t,k]$), and $\widetilde{\mathbf{r}_{i}}[t,k]$ is the identical result of the bit-level random response at this bit position. 
$(b)$ is because each of the last $K$ bits of entry $\mathbf{r}_{i}[t]$ (or $\mathbf{r}_{i}[t]'$) are changed   independently with probability $p$, and $(c)$ can be obtained by applying   $u\oplus v = (1-u)v+u(1-v)$ for any binary variable $u$ and $v$. Then, by making $\prod_{k=1}^{K}  \frac{1-p}{p}\leq \epsilon $, we complete the proof.
\end{proof}

\section{Proof of Theorem \ref{corollary:dpfp_share}}
\label{sec:dpfp_share_proof}
\begin{proof}
Since the value of $\mathcal{U}_j(s)$ (the $j$th random value generated by $\mathcal{U}$) is uniformly distributed for a seed $s$ \cite{boneh2017graduate}, we have $\Pr\Big(\mathcal{U}_1(s)\ \mathrm{mod}\ \floor{\frac{1}{2p}} =0\Big) = 1/\floor{\frac{1}{2p}} >2p$. Similarly,  $\Pr(x=0)=\frac{1}{2}$, thus, for any given fingerprint bit $f$, we also have $\Pr\Big(B=1, \mathcal{U}_1(s)\ \mathrm{mod}\ \floor{\frac{1}{2p}} =0\Big) \geq \frac{1}{2}2p = p$, which suggests that each $\mathbf{r}_i[t,k]$ will be changed (i.e., XORed by 1) with probability higher than $p$, and this satisfies the condition in Theorem \ref{thm:dpfp}.
\end{proof}

\section{The Fingerprint Extraction Algorithm}\label{sec:appendix_extraction}

Algorithm \ref{algo:extract} summarizes the main steps of extracting the fingerprint bit-strings from a leaked fingerprinted database (discussed in Section \ref{sec:extract_fp_leaked}). 

\begin{algorithm}[htb]
\small
\SetKwInOut{Input}{Input}
\SetKwInOut{Output}{Output}
\Input{The original   database $\mathbf{R}$, the leaked   database $\overline{\mathbf{R}}$, the Bernoulli distribution parameter  $p$, database owner's secret key $\mathcal{Y}$, pseudorandom number sequence generator $\mathcal{U}$, and a fingerprint template.}
\Output{The extracted fingerprint from the leaked database.}

Initialize $\mathbf{c}_0(l) = \mathbf{c}_1(l)=0, \forall l\in[1,L]$.



Construct the fingerprintable set $\overline{\mathcal{P}}$.

\ForAll{$\overline{\mathbf{r}_i}\in\overline{\mathcal{P}}$}{

Set pseudorandom seed $s = \{\mathcal{Y}|\mathbf{r}_i.PmyKey|t|k\}$, 


\If{$\mathcal{U}_1(s)\ \mathrm{mod}\  \floor{\frac{1}{2p}} = 0$}{

\resizebox{1\hsize}{!}{Set mask bit $x = 0$, if $\mathcal{U}_2(s)$ is even; otherwise    $x=1$.}

Set fingerprint index $l=\mathcal{U}_3(s)\ \mathrm{mod}\ L$.

Recover mark bit $B = \overline{\mathbf{r}_i}[t,k]\oplus \mathbf{r}_i[t,k]$. 

Recover fingerprint bit $f_l = x\oplus B$.

$\mathbf{c}_1(l)++$, if $f_l=1$; otherwise $\mathbf{c}_0(l)++$.


}

}

\ForAll{$l\in[1,L]$}{

$\mathbf{f}(l)=1$, if $\mathbf{c}_1(l) > \mathbf{c}_0(l)$; otherwise, $\mathbf{f}(l)=0$.\label{line:majority} 
}

 Return the extracted fingerprint bit string $\mathbf{f}$.

		\caption{Fingerprint extraction procedure}
\label{algo:extract}
\end{algorithm}

\section{Proof of Proposition \ref{thm:cell_accuracy}}\label{sec:proof_cell_accuracy}
\begin{proof} The fingerprinting mechanism only changes the last $K$ bits of selected data entries, thus, we have

$\mathbb{E}\Big[\Big| \widetilde{\mathbf{r}_i[t]}-\mathbf{r}_i[t] \Big|\Big] \\
    =   \mathbb{E}\Big[\Big| \sum_{k=1}^K\widetilde{\mathbf{r}_i[t,k]}2^{\widetilde{\mathbf{r}_i[t,k]}-1}-\mathbf{r}_i[t,k]2^{\mathbf{r}_i[t,k]-1}\Big| \Big] \\ 
      =   \mathbb{E}\Big[\Big| \sum_{k=1}^K   (\mathbf{r}_i[t,k]\oplus B)2^{\mathbf{r}_i[t,k]\oplus B-1}-\mathbf{r}_i[t,k]2^{\mathbf{r}_i[t,k]-1} \Big|\Big] \\
    =  \resizebox{1\hsize}{!}{$  \mathbb{E}\Big[ \Big|\sum_{k=1}^K   (\mathbf{r}_i[t,k] +  B-  2\mathbf{r}_i[t,k] B)2^{\mathbf{r}_i[t,k] +  B-  2\mathbf{r}_i[t,k] B-1}-\mathbf{r}_i[t,k]2^{\mathbf{r}_i[t,k]-1} \Big|\Big]$} \\
     = \Big| \sum_{k=1}^K  \Big( (  1-  \mathbf{r}_i[t,k] )2^{ -  \mathbf{r}_i[t,k] }-\mathbf{r}_i[t,k]2^{\mathbf{r}_i[t,k]-1} \Big)\Big| p \\
     = \Big| \sum_{k=1}^K  \Big( \overline{ \mathbf{r}_i[t,k]}2^{ \overline{\mathbf{r}_i[t,k]}-1 }-\mathbf{r}_i[t,k]2^{\mathbf{r}_i[t,k]-1} \Big) \Big| p.\\$
    Since  $\sum_{k=1}^K   \overline{ \mathbf{r}_i[t,k]}2^{ \overline{\mathbf{r}_i[t,k]}-1 }$ is the decimal representation of the complement of the last $K$ bits of $\mathbf{r}_i$, and according to Definition \ref{def:bit_sensitivity}, $\sum_{k=1}^K  \Big( \overline{ \mathbf{r}_i[t,k]}2^{ \overline{\mathbf{r}_i[t,k]}-1 }-\mathbf{r}_i[t,k]2^{\mathbf{r}_i[t,k]-1} \Big)$ falls in the range of $[-\Delta, \Delta]$, so its absolute value falls in $[0,\Delta]$, which completes the proof. 
\end{proof}

\section{Proof of Proposition \ref{prop:att_inf_bound}}\label{sec:proof_att_inf_bound}
\begin{proof}[Proof of Proposition \ref{prop:att_inf_bound}]
Mathematically, we have

$\mathrm{InfCap}= \Pr\left(\boldsymbol{r}_i[t]=\zeta_1\Big|\mathcal{M}(\R),\R_{/\boldsymbol{r}_i[t]}\right)\\
 =\frac{  \Pr(\mathcal{M}(\R) | \boldsymbol{r}_i[t]=\zeta_1  ,\R_{/\boldsymbol{r}_i[t]})  \Pr(\boldsymbol{r}_i[t]=\zeta_1  ,\R_{/\boldsymbol{r}_i[t]}) }{\Pr(\mathcal{M}(\R),\R_{/\boldsymbol{r}_i[t]})}  \\
 =   \frac{\Pr(\mathcal{M}(\R) | \boldsymbol{r}_i[t]=\zeta_1  ,\R_{/\boldsymbol{r}_i[t]}) }{\Pr(\mathcal{M}(\R)|\R_{/\boldsymbol{r}_i[t]})}   \Pr(\boldsymbol{r}_i[t]=\zeta_1  |\R_{/\boldsymbol{r}_i[t]}) \\
 =   \boxed{\frac{\Pr(\mathcal{M}(\R) | \boldsymbol{r}_i[t]=\zeta_1  ,\R_{/\boldsymbol{r}_i[t]}) }{\Pr(\mathcal{M}(\R)|     \boldsymbol{r}_i[t]=\zeta_2,    \R_{/\boldsymbol{r}_i[t]})}   }   
\frac{\Pr(\mathcal{M}(\R)|     \boldsymbol{r}_i[t]=\zeta_2,    \R_{/\boldsymbol{r}_i[t]})}{\Pr(\mathcal{M}(\R)|     \R_{/\boldsymbol{r}_i[t]})} \\
\times\Pr(\boldsymbol{r}_i[t]=\zeta_1  |\R_{/\boldsymbol{r}_i[t]})\\
 \stackrel{(a)}=     \boxed{e^\epsilon}     \frac{  \Pr(\mathcal{M}(\R),     \boldsymbol{r}_i[t]=\zeta_2,    \R_{/\boldsymbol{r}_i[t]})    \Pr(\R_{/\boldsymbol{r}_i[t]})     }{ \Pr(\boldsymbol{r}_i[t]=\zeta_2,    \R_{/\boldsymbol{r}_i[t]}) \Pr(\mathcal{M}(\R),     \R_{/\boldsymbol{r}_i[t]})      } \\
 \times \Pr(\boldsymbol{r}_i[t]=\zeta_1  |\R_{/\boldsymbol{r}_i[t]})\\
 =     e^\epsilon \boxed{\Pr(\boldsymbol{r}_i[t]=\zeta_2|\mathcal{M}(\R),\R_{/\boldsymbol{r}_i[t]})} \frac{\Pr(\boldsymbol{r}_i[t]=\zeta_1  |\R_{/\boldsymbol{r}_i[t]})}{\Pr(\boldsymbol{r}_i[t]=\zeta_2|    \R_{/\boldsymbol{r}_i[t]})}\\
 \stackrel{(b)}  \leq  e^\epsilon \boxed{(1-\mathrm{ InfCap})} \psi,$
which can be further simplified as  $\mathrm{ InfCap} \leq  \frac{ \psi e^\epsilon}{ \psi e^\epsilon +1}$.  Note that at line $(a)$, we apply the definition of $\epsilon$-entry-level differential privacy. At line $(b)$,  $ \psi = \frac{\Pr(\boldsymbol{r}_i[t]  |\R_{/\boldsymbol{r}_i[t]})}{\Pr(\boldsymbol{r}_i[t]'|    \R_{/\boldsymbol{r}_i[t]})}$ is the ratio between prior probabilities. 
\end{proof}

\section{Proof of Proposition \ref{thm:joint}}
\begin{proof}\label{sec:proof_joint}

We let binary vector $\mathbf{\tau}_t\in\{0,1\}^K$  (or $\mathbf{\tau}_z\in\{0,1\}^K$) represent the $K$ mark bits embedded to entries of the $t$th (or $z$th) attribute. 
$||\cdot||_0$ denotes the $l_0$ norm, which counts the number of nonzero entries in a vector. $\pi_{[2]}$ and $\omega_{[2]}$ are the binary representations of $\pi$ and $\omega$, respectively. Then, we have
\begin{gather*}
    \begin{aligned}
   &     \Pr(\widetilde{\R[t]}=\pi,\widetilde{\R[z]}=\omega)\\
= & \sum_{\mathbf{\tau}_t\in\{0,1\}^K} \sum_{\mathbf{\tau}_z\in\{0,1\}^K}  \Pr\Big(\R[t] = \pi_{[2]}\oplus \mathbf{\tau}_t,\R[z] = \omega_{[2]}\oplus \mathbf{\tau}_z\Big)\\
& \times p^{||\mathbf{\tau}_t||_0}(1-p)^{K-||\mathbf{\tau}_t||_0}p^{||\mathbf{\tau}_z||_0}(1-p)^{K-||\mathbf{\tau}_z||_0}\\
= & \Pr\Big(\R[t] = \pi_{[2]},\R[z] = \omega_{[2]}\Big) (1-p)^{2K}\\
& + \resizebox{0.95\hsize}{!}{$\sum_{\mathbf{\tau}_t\in\{0,1\}^K/\mathbf{0}} \sum_{\mathbf{\tau}_z\in\{0,1\}^K/\mathbf{0}}  \Pr\Big(\R[t] = \pi_{[2]}\oplus \mathbf{\tau}_t,\R[z] = \omega_{[2]}\oplus \mathbf{\tau}_z\Big)$}\\
& \times p^{||\mathbf{\tau}_t||_0}(1-p)^{K-||\mathbf{\tau}_t||_0}p^{||\mathbf{\tau}_z||_0}(1-p)^{K-||\mathbf{\tau}_z||_0}.
    \end{aligned}
\end{gather*}

By denoting the second summand in the above equation as $\diamondsuit$, we have\\
$\diamondsuit \leq  \Pr_{\max}(\R[t],\R[z])\times \Big( \sum_{\mathbf{\tau}_t\in\{0,1\}^K/\mathbf{0}} \sum_{\mathbf{\tau}_z\in\{0,1\}^K/\mathbf{0}}\\ p^{||\mathbf{\tau}_t||_0}(1-p)^{K-||\mathbf{\tau}_t||_0}p^{||\mathbf{\tau}_z||_0}(1-p)^{K-||\mathbf{\tau}_z||_0}\Big)\\
 =\resizebox{0.95\hsize}{!}{$\Pr_{\max}(\R[t],\R[z])\times \Big( \sum_{\mathbf{\tau}_t\in\{0,1\}^K/\mathbf{0}}p^{||\mathbf{\tau}_t||_0}(1-p)^{K-||\mathbf{\tau}_t||_0}\Big)$} \\ \times  \Big(\sum_{\mathbf{\tau}_z\in\{0,1\}^K/\mathbf{0}}p^{||\mathbf{\tau}_z||_0}(1-p)^{K-||\mathbf{\tau}_z||_0}\Big) \\
  = \Pr_{\max}(\R[t],\R[z])\Big(1-(1-p)^K    \Big)^2$.\\
Similarly, we also have \resizebox{0.6\hsize}{!}{$\diamondsuit\geq \Pr_{\min}(\R[t],\R[z])\Big(1-(1-p)^K    \Big)^2$}, thus, the proof is completed.
\end{proof}

\section{Analysis of $P_\mathrm{rbst\_rnd}$ in Random Bit Flipping Attack}\label{sec:analysis_rnd}

First, we show how to determine $D$ (number of bit matches with the malicious SP's fingerprint) given $C$ (number of times a database can be shared), and $L$ (length of fingerprinting string). If the database owner   shares its database with $C$ different SPs. To make the extracted fingerprint have the most bit matches with the  malicious SP, it requires that the probability of having more than $D$ bit matches is higher than $1/C$, i.e., ${L \choose D}(\frac{1}{2})^D(\frac{1}{2})^{L-D}\geq \frac{1}{C}$, which can be solved analytically.

In order to show  $P_\mathrm{rbst\_rnd}$ is monotonically increasing with $p$, we  define function \resizebox{0.65\hsize}{!}{$f(m) = \sum_{w_l\in\mathcal{W},l\in[1,L]}\sum_{\mathcal{L}_D}\prod_{l\in\mathcal{L}_D} p_l(\frac{1}{L})^{w_l}$} ($m$ is embedded as a parameter of set $\mathcal{W}$, i.e., \resizebox{0.6\hsize}{!}{$\mathcal{W} = \{w_1,w_2,\cdots,w_L>0|\sum_{l=1}^Lw_l=m\}$}). Then, $P_{\mathrm{rbst\_rnd}}$ represents the expected value of $f(m), m\sim\mathrm{Binomial}(NKT,2p)$. As a result, it is sufficient to show that $f(m)$ is monotonically increasing with $m$. First, we observe  that $p_l$ ($p_l = \sum_{q=0}^{\floor{\frac{w_l}{2}}}{w_l \choose q}\gamma_{\mathrm{rnd}}^q(1-\gamma_{\mathrm{rnd}})^{w_l-q}$) is the cumulative distribution function  of a binomial distribution  function, which is monotonically increasing with $w_l$, thus the multiplication of all $p_l$'s, i.e., $\prod_{l\in\mathcal{L}_D}p_l$ is increasing with $m=\sum_lw_l$. Second, it is easy to check that the carnality of $\mathcal{W}$ is $m!S(w,L)$, where $S(w,L)$ represents Stirling number of the second kind (i.e., the number of ways to partition a set of $w$ objects into $L$ non-empty subsets) \cite{graham1989concrete}. Since $m!$ grows faster then $L^{w_l}$ as $m$ increases,\footnote{In real-life applications, we have $m\gg L\gg \ln N$.} we can conclude that $f(m)$ is monotonically increasing with $m$ (the number of fingerprinted bit positions). When $0<p<0.5$, $P_{\mathrm{rbst\_rnd}}$ can be characterized as the summation of monotonically increasing functions with respect to $m$ and $p$,  which suggests that the higher the value of $p$, the more robust of the proposed fingerprinting mechanism is against the random bit flipping attack.

\section{Analysis of $G$ in the Correlation Attack}\label{sec:G_vs_p}
As per proposition \ref{thm:joint}, $|\Pr(  \widetilde{ \R[t]}=\pi, \widetilde{\R[z]} = \omega  )-\Pr(\R[t] = \pi,\R[z] = \omega)|$ falls in the range of $[0,\max\{|A|,|B|\}]$, where $A$ and $B$, respectively,  are\\
$\resizebox{1\hsize}{!}{$A =\Pr\Big(\R[t] = \pi,\R[z] = \omega\Big)\Big((1-p)^{2K}-1\Big)+\Pr_{\min}(\R[t],\R[z])\Big(1-(1-p)^K    \Big)^2 $}$,   $\resizebox{1\hsize}{!}{$B = \Pr\Big(\R[t] = \pi,\R[z] = \omega\Big)\Big((1-p)^{2K}-1\Big)+\Pr_{\max}(\R[t],\R[z])\Big(1-(1-p)^K    \Big)^2$}$. Without any assumption on the database, we consider each of the point $\tau$ (threshold) in $[0,\max\{|A|,|B|\}]$ has equal probability density \cite{ji2021curse}. Thus,  
$G = \frac{1-\prod_{z\in[1,T],z\neq t}\prod_{\omega}\frac{\tau}{\max\{|A|,|B|\}}}{(1-(1-p)^K)\Pr(\R[t]=\pi)}$. Let $\lambda = 1-(1-p)^K\in[0,1-(\frac{1}{2})^K]$, then, we have
\\
\resizebox{1\hsize}{!}{$A = \Pr\Big(\R[t] = \pi,\R[z] = \omega\Big)\Big((1-p)^{2K}+1\Big)(-\lambda)+\Pr_{\min}(\R[t],\R[z])\lambda^2$},\\
\resizebox{1\hsize}{!}{$B = \Pr\Big(\R[t] = \pi,\R[z] = \omega\Big)\Big((1-p)^{2K}+1\Big)(-\lambda)+\Pr_{\max}(\R[t],\R[z])\lambda^2$}. Thus, 
$G = \frac{1-(\frac{\tau}{\mathcal{O}(\lambda)})^{\sum_{z\in[1,T],z\neq t}k_z}}{\mathcal{O}(\lambda)}$, and  $k_z$ is the number of possible instances of attribute $z$. Since both the numerator and denominator increases with $p$, but the denominator grows with a much higher rate, $G$ decreases as $p$ increases.




\section{Proof of Theorem \ref{thm:proof_qualified}}\label{sec:proof_qualified_section}
\begin{proof}[Proof]
Suppose that  Algorithm \ref{algo:DeterminingQualifiedSP}  terminates with $l$ outputs  (it takes $l$ tries to determine $ID_{\mathrm{internal}}^c$, leading to a ``TRUE'' condition for the noisy comparison),   i.e., $\mathbf{a} = [a_1,a_2,\cdots,a_l] = \{\bot\}^{l-1}\cup \{\top\}$. 
By defining   
\begin{equation*}
    \begin{aligned}
        &f_i(\R,z) = \Pr(||\mathcal{M}_i(\R)-\R||_{1,1} + \mu_i< \Gamma+z_i),\\
         &g_i(\R,z) = \Pr(||\mathcal{M}_i(\R)-\R||_{1,1} +\mu_i \geq \Gamma+z_i),
    \end{aligned}
\end{equation*}
where $z_i$ is an instance of $\rho_i$ generated at Line \ref{line:generate_rho} in Algorithm \ref{algo:DeterminingQualifiedSP}. Then, we can  have\\
\begin{gather*}
    \begin{aligned}
     &  \frac{\Pr\Big(\mathrm{DetermineTheInternalIDforOneSP}(\R) = \mathbf{a}\Big)}{\Pr\Big(\mathrm{DetermineTheInternalIDforOneSP}(\R') = \mathbf{a}\Big)}\\
  =     & \frac{\int_{-\infty}^{\infty}\Pr(\rho_i=z_i)\prod_{i=1}^{l-1}f_i(\R,z_i) g_l(\R,z_i)dz_i}{\int_{-\infty}^{\infty}\Pr(\rho_i=z_i)\prod_{i=1}^{l-1}f_i(\R',z_i) g_l(\R',z_i)dz_i}\\
 \stackrel{(*)}= & \frac{\int_{-\infty}^{\infty}\Pr(\rho_i=z_i-\Delta)\prod_{i=1}^{l-1}f_i(\R,z_i-\Delta) g_l(\R,z_i-\Delta)dz_i}{\int_{-\infty}^{\infty}\Pr(\rho_i=z_i)\prod_{i=1}^{l-1}f_i(\R',z_i) g_l(\R',z_i)dz_i}\\=&\clubsuit,
    \end{aligned}
\end{gather*}
where  $(*)$ is obtained by changing all the integration variables, i.e.,  $z_i$'s, to $(z_i-\Delta)$'s, $\forall i\in\{1,2,3,\cdots\}$. Next, we investigate the three parts of the integrand in the numerator of $\clubsuit$ separately. 

First, we have $\Pr(\rho_i = z_i-\Delta)\leq e^{\epsilon_3} \Pr(\rho_i=z_i)$, as $\rho_i$ is attributed to a Laplace distribution whose parameter is calibrated using $\Delta$. 

Second, suppose $\R$ and $\R'$ differs at $r_{ij}$ and $r_{ij}'$. Then, \resizebox{1\hsize}{!}{$||\mathcal{M}_i(\R')-\R'||_{1,1}-||\mathcal{M}_i(\R)-\R||_{1,1} = |\widetilde{r_{ij}'} - r_{ij}'|-|\widetilde{r_{ij}} - r_{ij}|\leq \Delta$}, where $\widetilde{r_{ij}}$ (or $\widetilde{r_{ij}'}$) is the fingerprinted version of $r_{ij}$ (or $r_{ij}'$). The equality follows from that for  any specific SP (which uniquely determines a  pseudorandom seed),  Algorithm \ref{algo:fingerprint_sp_n} will select exactly the same bit positions in both $\R$ and $\R'$ to insert the fingerprint, and  all selected bits except for the different entry between $\R$ and $\R'$ will also be replaced with the exact same bit values. The inequality is because both $|\widetilde{r_{ij}} - r_{ij}|$ and $|\widetilde{r_{ij}'} - r_{ij}'|$ are upper bounded by $\Delta$. 
As a result, for the second part  in $\clubsuit$, we   
obtain $f_i(\R,z_i-\Delta) =$ 
\resizebox{1\hsize}{!}{$\Pr\big(||\mathcal{M}_i(\R)-\R||_{1,1} + \mu_i < \Gamma+z_i-\Delta\big) \leq \Pr\big(||\mathcal{M}_i(\R')-\R'||_F +\mu_i-\Delta < \Gamma+z_i-\Delta\big)$} $=f_i(\R',z_i)$, and the inequality holds  since we replace $||\mathcal{M}_i(\R)-\R||_{1,1}$ by a smaller value, i.e., $||\mathcal{M}_i(\R')-\R'||_F-\Delta$, which decreases the probability. 

Third, since $\mu_i$ is a Laplace noise, which is also calibrated using $\Delta$, we have  \resizebox{0.7\hsize}{!}{$g_l(\R,z_i-\Delta) = \Pr\Big(||\mathcal{M}_i(\R)-\R||_{1,1}+\mu_i  \geq \Gamma+z_i-\Delta\Big)$} \resizebox{1\hsize}{!}{$\leq e^{\epsilon_2}\Pr\Big(||\mathcal{M}_i(\R')-\R'||_{1,1}+\mu_i -\Delta  \geq \Gamma+z_i-\Delta\Big)  = e^{\epsilon_2} g_l(\R',z_i)$.}  \resizebox{1\hsize}{!}{Hence, $\clubsuit\leq \frac{\int_{-\infty}^{\infty}e^{\epsilon_3} \Pr(\rho_i=z_i)\prod_{i=1}^{l-1}f_i(\R',z_i)  e^{\epsilon_2} g_l(\R',z_i)dz_i}{\int_{-\infty}^{\infty}\Pr(\rho_i=z_i)\prod_{i=1}^{l-1}f_i(\R',z_i) g_l(\R',z_i)dz_i}   = e^{\epsilon_2+\epsilon_3}$}.
\end{proof}

\section{Database Sharing with Multiple SPs}\label{sec:appendix_multiple_shares}

In Algorithm \ref{algo:MultipleFingerprintedDatabasesReleasing}, we summarizes the steps of sharing fingerprinted databases with $C$ SPs (discussed in Section \ref{sec:multiple_shares}).

\begin{algorithm}
\small
\SetKwInOut{Input}{Input}
\SetKwInOut{Output}{Output}
\caption{\resizebox{0.8\hsize}{!}{Share Fingerprinted Databases with $C$ SPs}}\label{algo:MultipleFingerprintedDatabasesReleasing}
\Input{Original database $\R$, fingerprinting scheme $\mathcal{M}$,   sequence number of   SPs, i.e., $\{1,2,\cdots,C\}$,   threshold $\Gamma$, and privacy budget $\epsilon$, $\epsilon_2$,  $\epsilon_3$ and $\delta'$.} 
\Output{$\{a_1,a_2,a_3,\cdots,a_C\}$.}

Set   $count=0$.


\ForAll{\fbox{$c$th SP $c\in\{1,2\cdots,C\}$}}{

\ForAll{$i\in\{1,2,3,\cdots\}$}{

Generate   an instance of internal ID for the $c$th SP via $ID_{\mathrm{internal}}^c = Hash(\mathcal{K}|c|i)$.

Generate   $\mathcal{M}_i^c(\R)$ by calling  Algorithm \ref{algo:fingerprint_sp_n} with $ID_{\mathrm{internal}}^c$ and 
privacy budget $\epsilon$. 

Sample $\mu_i\sim\mathrm{Lap}(\frac{\Delta}{\epsilon_2})$ and $\rho_i\sim\mathrm{Lap}(\frac{\Delta}{\epsilon_3})$.


\eIf{$||\mathcal{M}_i^c(\R)-\R||_{1,1} +\mu_i \geq \Gamma+ \rho_i$}{\fbox{Output $a_i = \mathcal{M}_i^c(\R)$}.\label{line:output_fpdp} 

}{Output $a_i = \bot$.
}

}

}

\end{algorithm}

\section{Proof of Theorem \ref{thm:multiple_share_is_dp}}\label{sec:multiple_share_is_dp_proof}
\begin{proof}[Proof]  
Algorithm \ref{algo:MultipleFingerprintedDatabasesReleasing} is the composition of $C$ rounds of Algorithm   \ref{algo:DeterminingQualifiedSP} together with $C$ rounds of Algorithm   \ref{algo:fingerprint_sp_n}. According to   Theorem \ref{thm:comp}, $C$ rounds of Algorithm   \ref{algo:DeterminingQualifiedSP} and $C$ rounds of Algorithm   \ref{algo:fingerprint_sp_n} are   $(\sqrt{2C\ln(\frac{1}{\delta'})}(\epsilon_2+\epsilon_3)+C(\epsilon_2+\epsilon_3)(e^{\epsilon_2+\epsilon_3}-1),\delta'$)-differentially private and $(\sqrt{2C\ln(\frac{1}{\delta'})}\epsilon +C\epsilon (e^{\epsilon }-1),\delta'$)-differentially private, respectively. Then, by simple composition of those two, we complete the proof.
\end{proof}

\section{Detailed Experiment Results: Tables and Plots}\label{sec:miss_exp}

\subsection{Pairwise Difference in Each Class}\label{sec:fraction_diff_result}


\begin{wraptable}{r}{.7\linewidth}
    \begin{minipage}{1\linewidth}
\begin{adjustbox}{width=1\textwidth}
\begin{tabular}{| *{6}{c|} }
    \hline
abs. diff.&   0   & 1 & 2 & 3 & 4\\
\hline
not\_recom & 40\% & 46.79\% & 7.08\% & 5.12\% & 1\%\\
recommend & 93.75\%  &6.25\% & 0 & 0 &0\\ 
very\_recom & 35.10\% &49.71\% & 10.66\%  & 4.39\% & 0.13\%\\
priority &36.04\%&44.65\%& 11.31\% &  5.07\%  & 2.93\%\\ 
spec\_prior&50.50\% & 37.19\%&  9.01\% &  3.29\% & 0\\
    \hline
\end{tabular}
  \end{adjustbox}
   \end{minipage}
       \captionsetup{font=footnotesize}
\caption{Fraction of pairwise absolute differences between instances of attributes.}
\label{table:pairwise_diff}
\end{wraptable} 
In Table \ref{table:pairwise_diff}, we summarize the results of  the  fraction  of  pairwise  absolute differences taking a specific value in each class. Note that this observation helps us to control the sensitivity as discussed in Section \ref{sec:setup}.

\subsection{Validating $P_{\mathrm{rbst_{rnd}}}$ Empirically}\label{sec:app_robust_rnd}
In Figure \ref{fig:robust_rnd}, we plot the numerical value of $P_{\mathrm{rbst\_rnd}}$ varies with different values of $p$ (reported in Table \ref{table:dpfp_vs_vannila}).

\begin{figure}[htb]
  \begin{center}
     \includegraphics[width= 0.8\columnwidth]{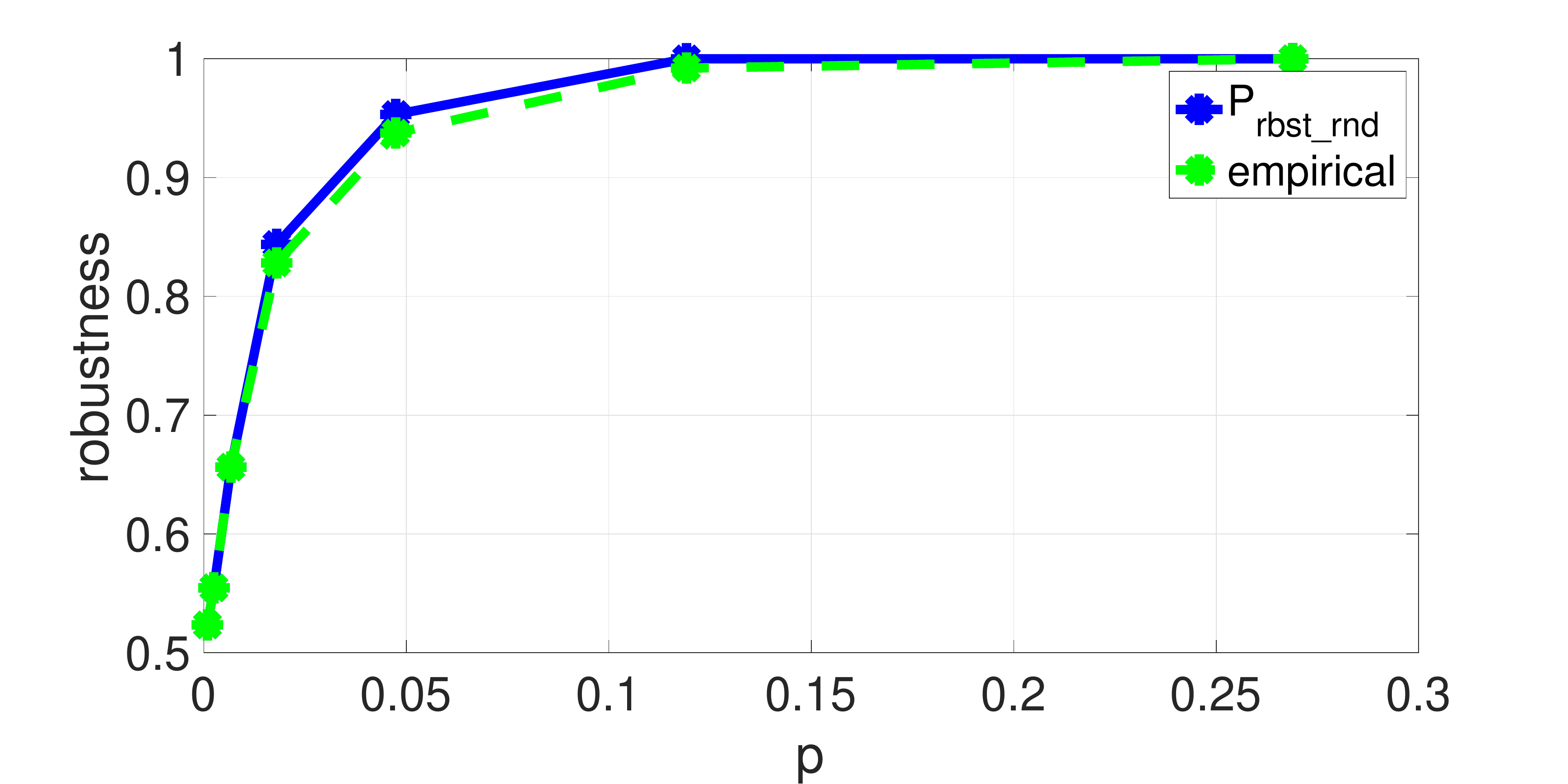}
      \end{center}
  \caption{\label{fig:robust_rnd} Theoretical and empirical values of  $P_{\mathrm{rbst\_rnd}}$   under different $p$.
}
\end{figure}

\subsection{Privacy Evaluation against Attribute Inference}\label{sec:app_att_inf_atk}

Here, we 
investigate the robustness of our proposed entry-level differentially-private fingerprinting scheme when it is subject to the attribute inference attack discussed in Section \ref{sec:privacy_against_att_inf}. Specifically, we aim to show that the attacker's inference capability for our considered database (in Section \ref{sec:setup}) is upper bounded by our theoretical findings   in  Proposition \ref{prop:att_inf_bound}. In particular, we consider the classical adversary in conventional differential privacy, which assumes data entries within a
dataset (database) are independent, i.e., data entries are mutually independent and identically distributed  (similar to    the ``independent tuple assumption'' adopted in   \cite{liu2016dependence}  - page 5 left column). Also, according to \cite{liu2016dependence}, under this adversary model, the auxiliary information $\R_{/\boldsymbol{r}_{i}[t]}$ (all the original entries in the database except for $\boldsymbol{r}_{i}[t]$) can be used to serve as sampling
values of $\boldsymbol{r}_{i}[t]$, which can be used to estimate the posterior 
probability $\Pr(\boldsymbol{r}_{i}[t] = \xi_1)$, where $\xi_1$ is any possible value that $\boldsymbol{r}_{i}[t]$ can take. 
Then, according to statistical inference,   the occurrence frequency of each value of $\xi_1$ is an unbiased estimator of $\Pr(\boldsymbol{r}_{i}[t] = \xi_1)$. Thus, the attacker's inference capability ($\mathrm{InfCap}$) defined in (\ref{eq:InfCap}) can be empirically  
computed from $\R_{/\boldsymbol{r}_{i}[t]}$ (note that \cite{liu2016dependence} applies the similar attack approach). 
In Figure \ref{fig:att_inf_atk_exp}, we compare the empirical results of $\mathrm{InfCap}$ 
for the inference of $\boldsymbol{r}_{i}[t]$ when $t$ (the attribute) is chosen as ``finance'' or ``housing'' with the theoretical upper bound. Clearly, the considered adversary cannot have higher inference capability than our derived upper bound. Note that we discuss how to deal with attribute inference attack that uses data correlation in Section \ref{sec:diss_att_corr} (we will further study it in future work). 
\begin{figure}[htb]
  \begin{center}
     \includegraphics[width= 0.8\columnwidth]{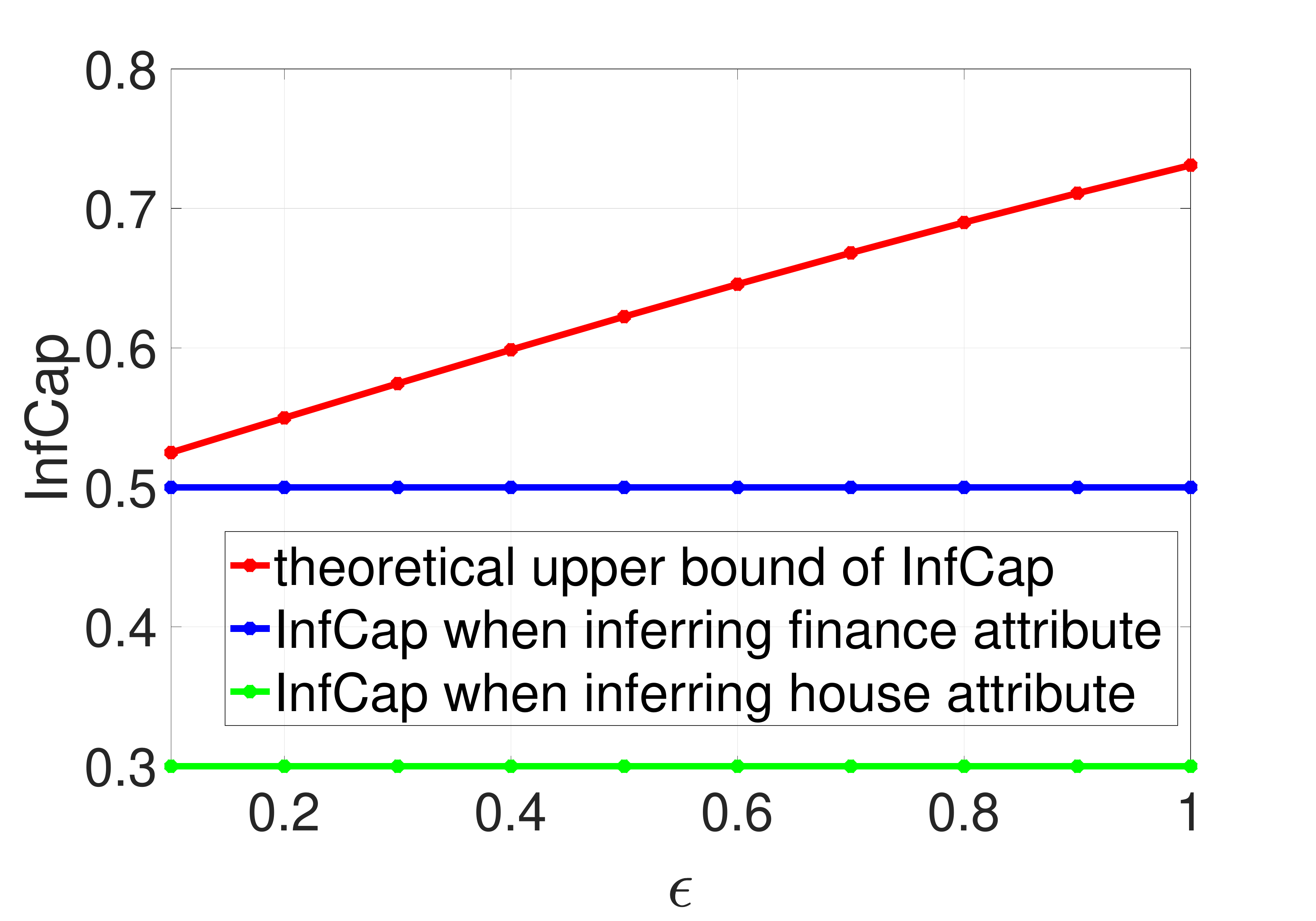}
      \end{center}
  \caption{\label{fig:att_inf_atk_exp} Empirical inference capability of the attacker v.s.   upper bound. 
}
\end{figure}

\subsection{Experiment Results on Task-independent Database Utility}\label{sec:exp_task_indpend}

In this experiment, we consider a lower privacy budget  (which provides a higher privacy guarantee and higher fingerprint robustness, as demonstrated in Figure \ref{fig:relationships}), and let  $\epsilon$ take values 0.25, 0.5, 0.75, and 1.  We do not consider moderate to high values of $\epsilon$ as in the previous section, because they return much higher utility values.  


\begin{table*}[htb]
\begin{center}
\resizebox{\columnwidth}{!}{
\begin{tabular}{| *{10}{c|} }
    \hline
 $\epsilon$&    $\infty$   & \multicolumn{2}{c|}{0.25} & \multicolumn{2}{c|}{0.5}
            & \multicolumn{2}{c|}{0.75}
                    & \multicolumn{2}{c|}{1} \\
    \hline
\% of changed entries &\cellcolor{Blue}  NA    &\cellcolor{Red}35.03\% &49.12\% & \cellcolor{Red}30.18\% &39.28\% &\cellcolor{Red}25.58\% &33.85\% &\cellcolor{Red}21.38\% &27.92\%   \\
Var. of 1st Att. &\cellcolor{Blue} 0.6667 &\cellcolor{Red}+0.0002 &+0.0392 & \cellcolor{Red}+0.0015 &+0.0268 &\cellcolor{Red}-0.0026 &+0.0152 &\cellcolor{Red}-0.0011 &+0.0122   \\ 
Var. of 2nd Att. &\cellcolor{Blue} 2.0002 &\cellcolor{Red}+0.0010 &+0.1644 & \cellcolor{Red}-0.0048 &+0.0424 &\cellcolor{Red}+0.0029 &+0.0178 &\cellcolor{Red}+0.0025 &+0.0231   \\ 
Var. of 3rd Att. &\cellcolor{Blue} 1.2501 &\cellcolor{Red}-0.0090 &+0.1023 & \cellcolor{Red}-0.0066 &+0.0632 &\cellcolor{Red}+0.0021 &+0.0432 &\cellcolor{Red}+0.0074 &+0.0154   \\ 
Var. of 4th Att. &\cellcolor{Blue} 1.2501 &\cellcolor{Red}-0.0128 &+0.1109 & \cellcolor{Red}-0.0167 &+0.0732 &\cellcolor{Red}+0.0009 &+0.0321 &\cellcolor{Red}-0.0003 &+0.0172   \\ 
Var. of 5th Att. &\cellcolor{Blue} 0.6667 &\cellcolor{Red}-0.0059 &+0.0412 & \cellcolor{Red}-0.0068 &+0.0321 &\cellcolor{Red}+0.0031 &+0.0178 &\cellcolor{Red}+0.0017 &+0.0092   \\ 
Var. of 6th Att. &\cellcolor{Blue} 0.2500 &\cellcolor{Red}-0.0000 &-0.0012 & \cellcolor{Red}-0.0000 &-0.0001 &\cellcolor{Red}-0.0000 &+0.0221 &\cellcolor{Red}-0.0000 &-0.0010   \\ 
Var. of 7th Att. &\cellcolor{Blue} 0.6667 &\cellcolor{Red}-0.0033 &+0.0432 & \cellcolor{Red}-0.0022 &0.0233 &\cellcolor{Red}-0.0005 &-0.0017 &\cellcolor{Red}-0.0007 &+0.0089   \\ 
Var. of 8th Att. &\cellcolor{Blue} 0.6667 &\cellcolor{Red}-0.0048 &+0.0321 & \cellcolor{Red}-0.0042 &0.0198 &\cellcolor{Red}-0.0039 &+0.0042 &\cellcolor{Red}-0.0011 &+0.0077   \\ 
    \hline
\end{tabular}
}
\end{center}
\caption{Comparison of database utility (measured in terms of the change of variance of each attribute in the database). Column colored in blue is the original variance of each attribute in $\R$. Columns colored in red are the change of variance caused by our proposed mechanism under various privacy budgets. Non-highlighted columns are the change of variance caused by applying local differentially-private perturbation  followed by fingerprinting proposed in \cite{li2005fingerprinting}.} 
\label{table:dpfp_vs_dpfollowedbyfp}
\end{table*}

We measure the general database utility as the change in the variance of each attribute caused by both mechanisms.   Table \ref{table:dpfp_vs_dpfollowedbyfp} 
summarizes the fraction of changed data entries and the change of variance under varying $\epsilon$. The column for $\epsilon = \infty$ (highlighted in blue) records the original variance of each attribute in $\R$, columns highlighted in red are the results achieved by our   mechanism, and the non-highlighted columns are the results obtained by applying  local  differentially-private perturbation  followed by fingerprinting (i.e., the two-stage approach we use for comparison). It is clear that given the same privacy guarantee and fingerprint robustness, our   mechanism is able to maintain higher database utility by making less changes in $\mathbf{R}$, i.e.,  the changes in variance caused by our mechanism are  generally 10 times smaller than that caused by the approach we use for comparison. This suggests that by merging differential privacy and fingerprinting as a unified mechanism, we can approximately  boost the statistical utility of the released nursery database by 10 times compared to the two-stage approach.

We also investigate the accuracy of the released database queries by considering the following customized queries: (i) $Q1$: the   data records which have more than 3 children (``more'') and the social conditions are  slightly  problematic (``slightly\_prob''), and (ii) $Q2$: the  data record whose parents' occupation are ``usual'' and the finance condition are inconvenient (``inconv''), i.e.,
\begin{equation*}
\begin{aligned}
    Q1: & \mathrm{SELECT\ PmyKey\ FROM\ Nursery\ WHERE } \\
    & \mathrm{children=more\ AND\ social=slightly\_prob}\\
        Q2: & \mathrm{SELECT\ PmyKey\ FROM\ Nursery\ WHERE } \\
    & \mathrm{parent=usual\ AND\ finance=incov}
\end{aligned}.
\end{equation*}
By varying $\epsilon$ from 0.25 to 2, we show the query accuracy (i.e., the fraction of matched data records with the query results obtained from the original database) achieved by our mechanism and that by the two-stage approach in Figure~\ref{fig:query_accuracy}. 
Again, we observe that our proposed mechanism can achieve more accurate database queries than the two-stage approach.

\begin{figure}[htb]
  \begin{center}
  \begin{tabular}{cc}
     \includegraphics[width= .45\columnwidth]{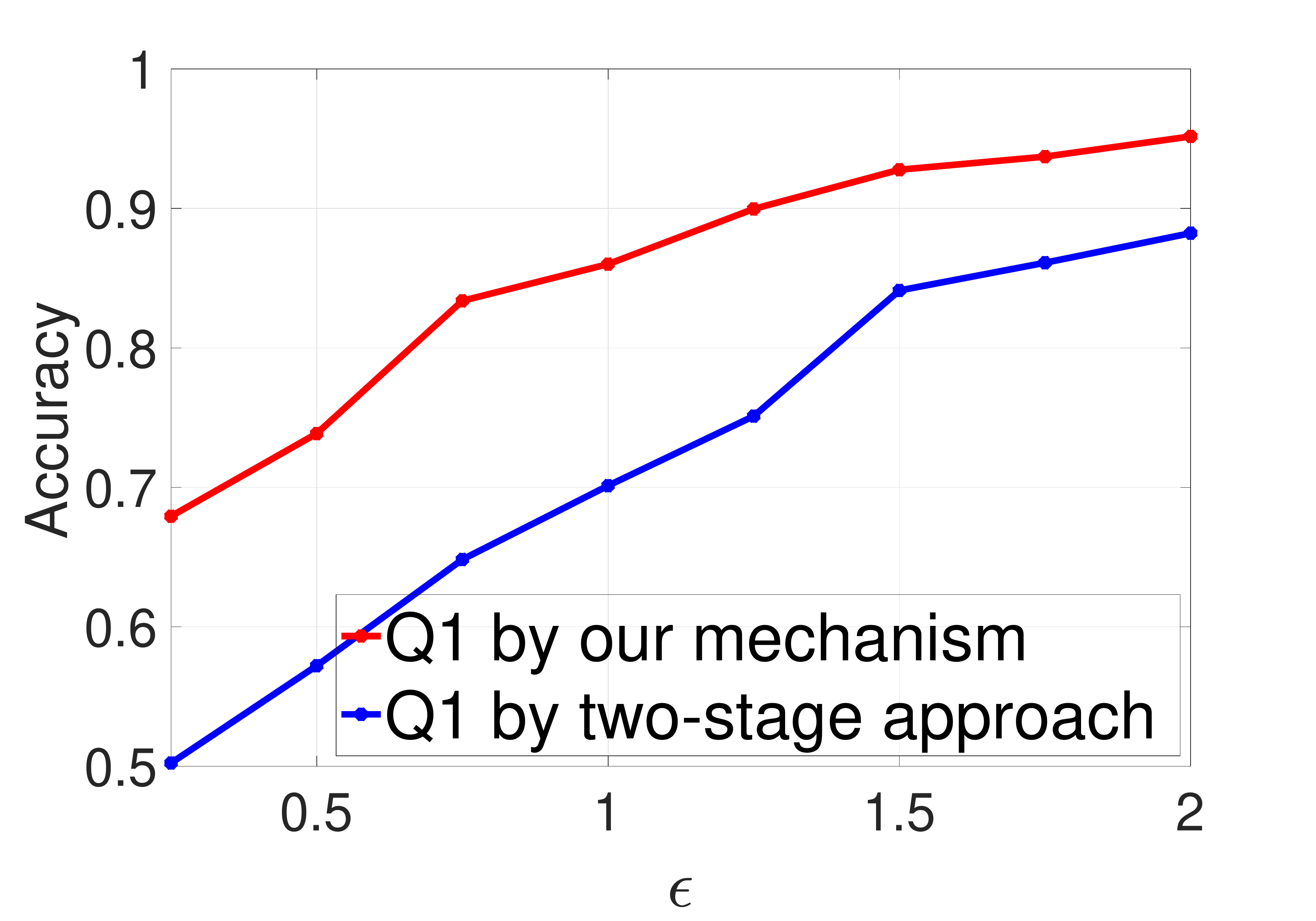}
		&
		     \includegraphics[width= .45\columnwidth]{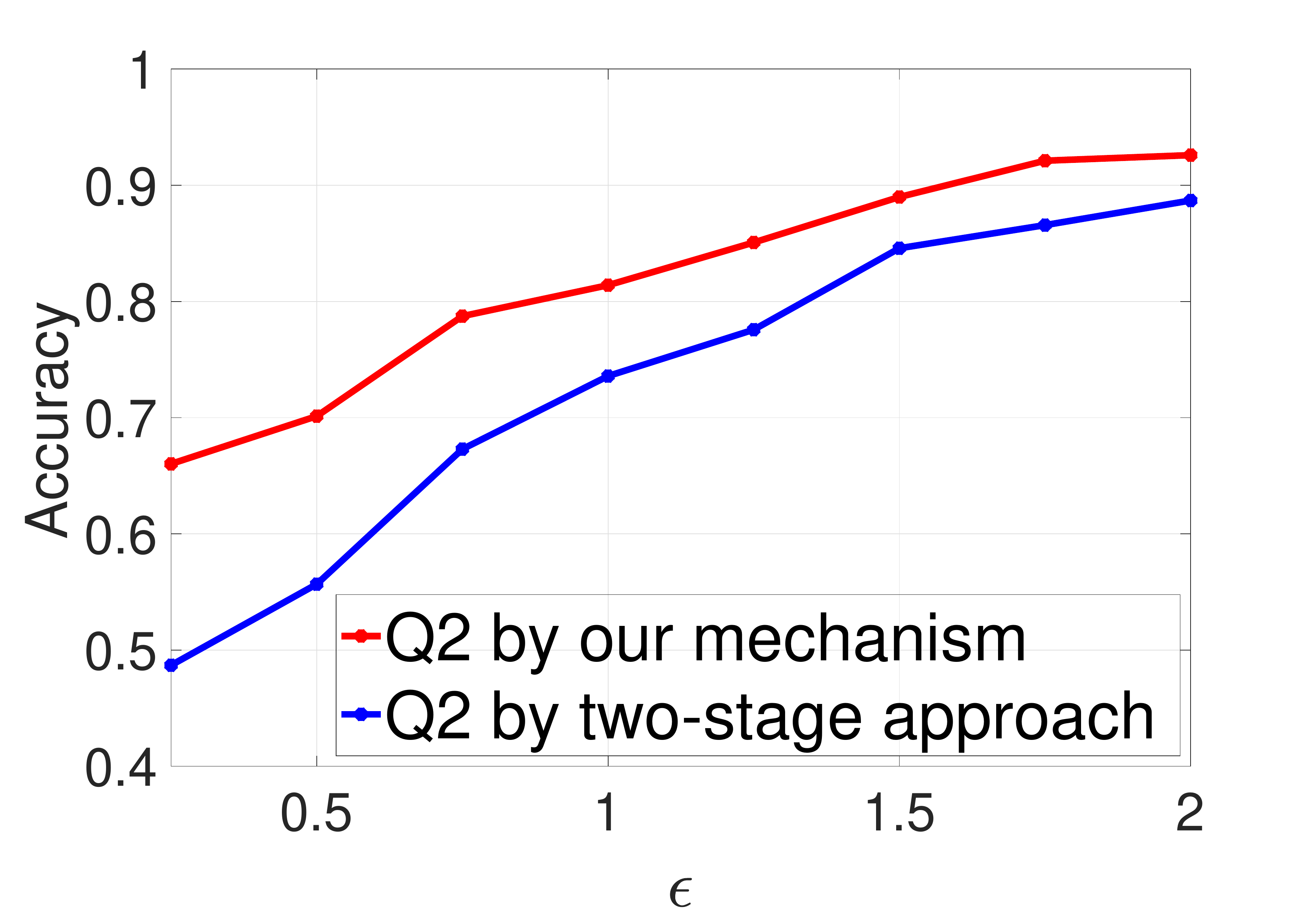}
 \\
    {\small  (a) Accuracy of $Q1$.}  &
    {\small  (b) Accuracy of $Q2$.} \\
        \end{tabular}
      \end{center}
  \caption{\label{fig:query_accuracy} Comparison of accuracy of the queries achieved by our mechanism and that by the two-stage approach.}
\end{figure}

\subsection{Experiment Results on Task-specific Database Utility}\label{sec:app_exp_task_specific}
In Figure \ref{fig:task-specific}, we show the fingerprinted database utility in the considered task-specific  applications. 
\begin{figure}[htb]
  \begin{center}
  \begin{tabular}{cc}
     \includegraphics[width= .35\columnwidth]{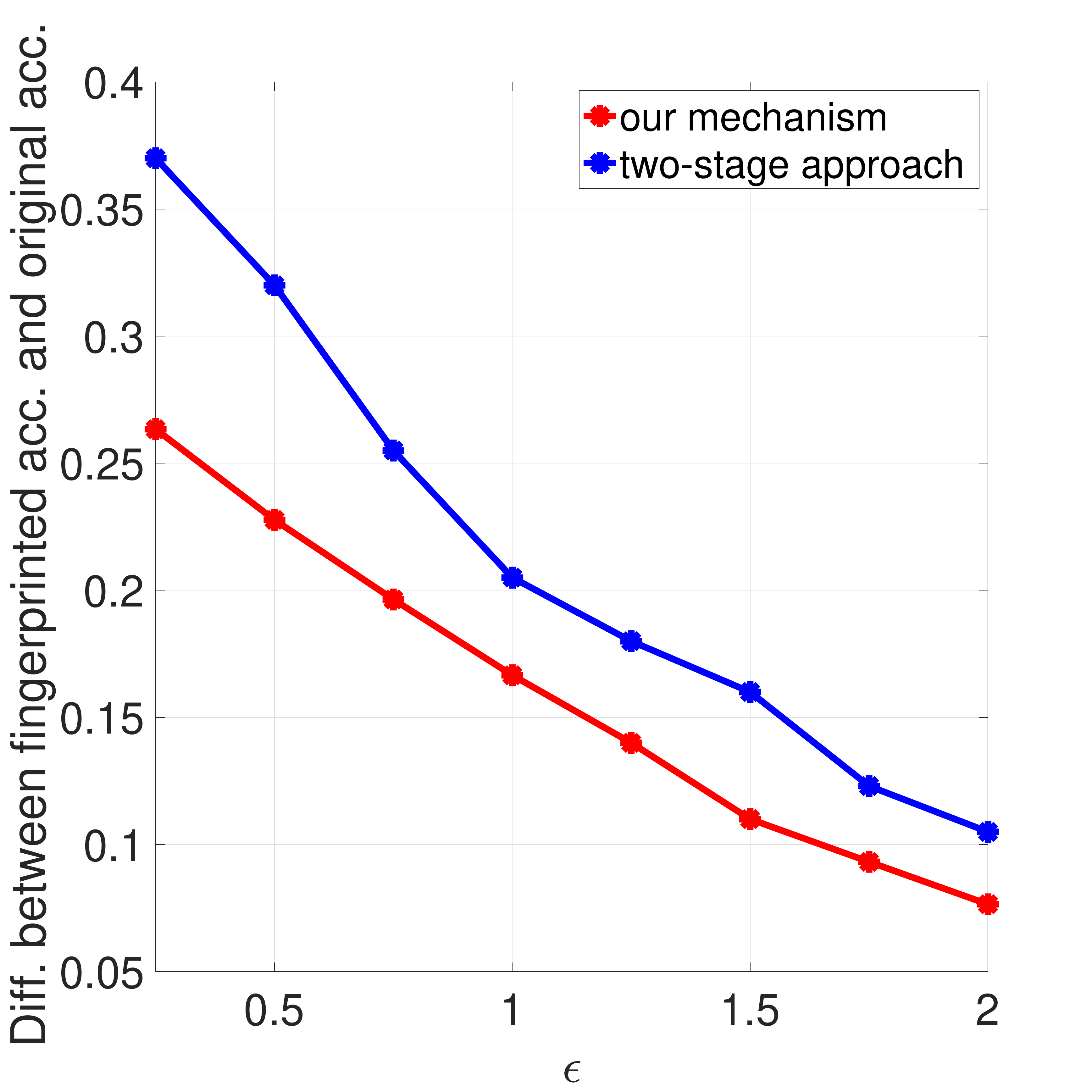}
&   \includegraphics[width= .35\columnwidth]{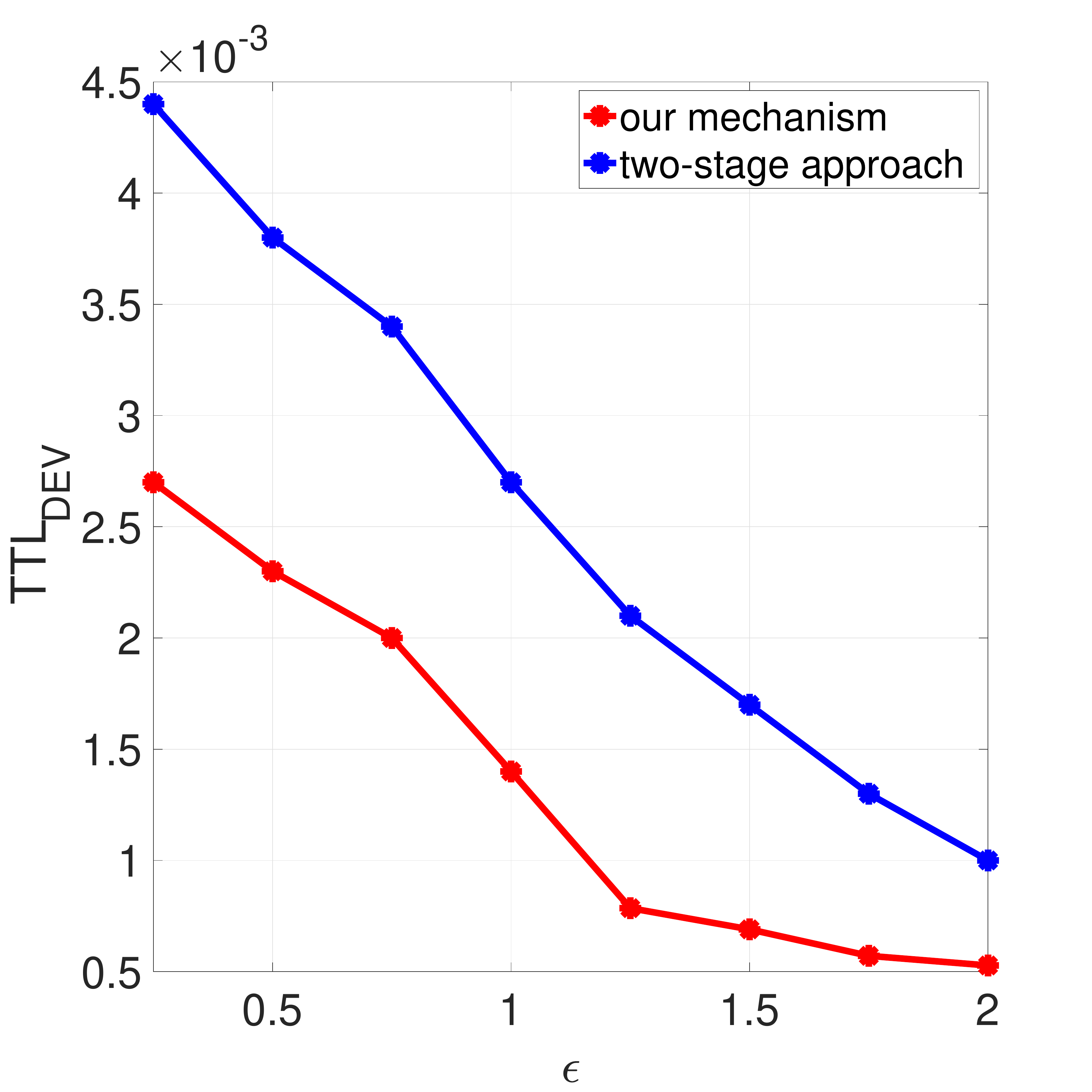}\\
	{\small  (a)   SVM Classification.}&
    {\small  (b)   PCA.} 
    \end{tabular}
      \end{center}
  \caption{\label{fig:task-specific}Database utility   in task-specific applications. (a) SVM: difference between fingerprinted testing accuracy and original testing accuracy versus $\epsilon$. (b) $\mathrm{TTL_{DEV}}$ versus $\epsilon$.}
\end{figure}

\subsection{Experiment Results on Fingerprint  Densities Comparison under Multiple Sharings}\label{sec:app_multiple_sahres}
In Figure \ref{fig:density_plot}, we plot the fingerprint densities of the shared databases when $\epsilon_2 = \epsilon_3$, i.e., $||\mathcal{M}(\R)-\R||_{1,1}$ as a result of sharing with 100 SPs.
\begin{figure}[htb]
  \begin{center}
     \includegraphics[width= 0.8\columnwidth]{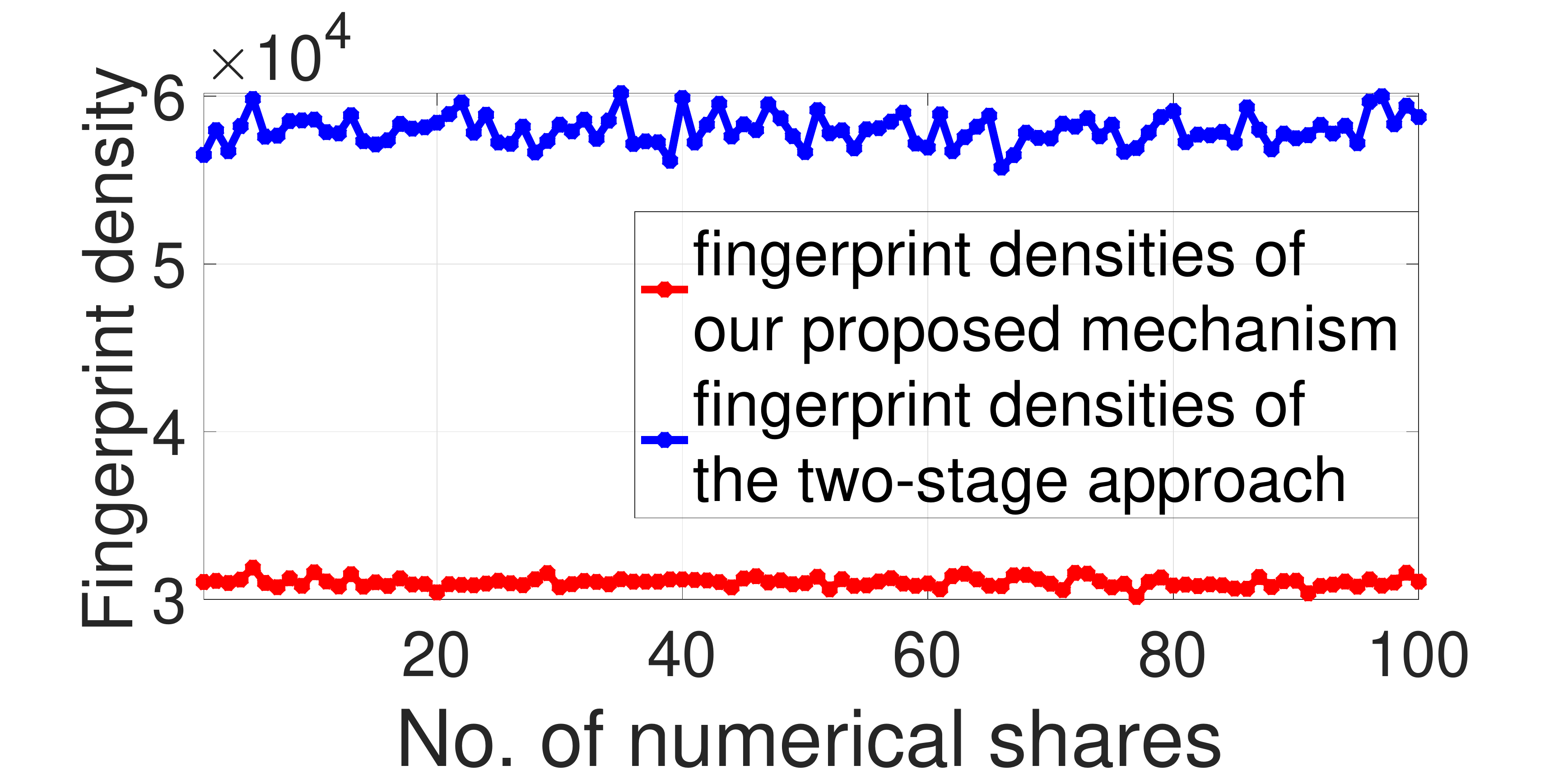}
      \end{center}
  \caption{\label{fig:density_plot}  Comparison of fingerprint densities of 100 fingerprinted databases achieved by the proposed mechanism and the two-stage approach.  
}
\end{figure}